\documentclass[a4paper,final]{llncs}

\usepackage{amsmath,amssymb}
\usepackage[final]{graphicx}
\usepackage{algorithm}
\usepackage{algorithmic}
\usepackage{fullpage}

\title{Approximating Semi-Matchings in Streaming and in Two-Party Communication}
\author{Christian Konrad\inst{1} and Adi Ros\'en\inst{2} }
\institute{LIAFA, Universit\'{e} Paris Diderot, France. \\ \email{konrad@lri.fr} \and CNRS and Univerit\'e Paris Diderot - Paris 7, France. \\ \email{adiro@liafa.univ-paris-diderot.fr}}

\DeclareMathOperator*{\degmax}{deg\,max}
\DeclareMathOperator*{\argmin}{arg\,min}

\DeclareMathOperator*{\polylog}{polylog}

\newcommand{\semi}{\mathrm{semi}}
\newcommand{\isemi}{\mathrm{isemi}}
\newcommand{\ALGOincomplete}{\textsc{incomplete}}
\newcommand{\ALGOsemi}{\textsc{asemi}}
\newcommand{\Order}{\mathrm{O}}
\newcommand{\OrderT}{\tilde{\mathrm{O}}}

\begin{document}
\maketitle 
\begin{abstract}
We study the communication complexity and streaming complexity of approximating unweighted semi-matchings.
A semi-matching in a bipartite graph $G = (A, B, E)$, with $n = |A|$, is a subset of edges $S \subseteq E$ that matches
all $A$ vertices to $B$ vertices with the goal usually being to %minimize the maximal number of $A$ vertices that
%are matched to the same $B$ vertex. 
do this as fairly as possible. While the term \textit{semi-matching} was coined in 2003 by Harvey et al. [WADS 2003], 
the problem had already previously been studied in the scheduling literature under different names.

We present a deterministic one-pass streaming algorithm that for any $0 \le \epsilon \le 1$ uses space 
$\OrderT(n^{1+\epsilon})$ and computes an $\Order(n^{(1-\epsilon)/2})$-approximation to the semi-matching problem. 
Furthermore, with $\Order(\log n)$ passes it is possible to compute an $\Order(\log n)$-approximation with space $\OrderT(n)$.

In the one-way two-party communication setting, we show that
for every $\epsilon > 0$, deterministic communication protocols for computing 
an $\Order(n^{\frac{1}{(1+\epsilon)c + 1}})$-approximation require a message of size 
more than  $cn$ bits. We present two deterministic protocols communicating $n$ and $2n$ edges that compute an $\Order(\sqrt{n})$ 
and an $\Order(n^{1/3})$-approximation respectively.

Finally, we improve on results of Harvey et al. [Journal of Algorithms 2006] and prove new links between 
semi-matchings and matchings. %As a by-product, we obtain a new and purely structural proof of the 
%$\lceil \log(n + 1) \rceil$ competitive ratio of the deterministic greedy online semi-matching algorithm of Azar et al. 
%[Journal of Algorithms 1995].
While it was known that an optimal semi-matching contains a maximum matching, we show that there is a hierachical decomposition 
of an optimal semi-matching into maximum matchings. A similar result holds for semi-matchings that do not admit length-two
degree-minimizing paths.
\end{abstract}

\section{Introduction}

\textbf{Semi-Matchings.}
A \textit{matching} in an unweighted bipartite graph $G = (A, B, E)$ can be seen as a one-to-one assignment matching
the $A$ vertices to $B$ vertices. The usual aim
is to find a matching that leaves as few $A$ vertices without associations as possible. A \textit{semi-matching} 
is then an extension of a matching, in that it is required that \textit{all} $A$ vertices are matched to 
$B$ vertices. This, however, is generally not possible in an injective way, and therefore we now allow the matching of multiple 
$A$ vertices to the same $B$ vertex. Typical objectives here are to minimize the maximal number of $A$ vertices that are matched to the 
same $B$ vertex, or to optimize with respect to even stronger balancing constraints. 
The term 'semi-matching' was coined by \cite{hllt03} and also used in \cite{fln2010}, however, the problem had
already previously been intensely studied in the scheduling literature \cite{ecs73,h73,anr95,a03,ll04}. We stick to this
term since it nicely reflects the structural property of entirely matching one bipartition of the graph.

The most prominent
application of the semi-matching problem is that of assigning a set of unit-length jobs to a set of
identical machines with respect to assignment conditions expressed through edges between the two sets. 
The objective of minimizing the 
maximal number of jobs that a machine receives then corresponds to minimizing the \textit{makespan} of the scheduling
problem. Optimizing the cost function $\sum_{b \in B} \deg_S(b) (\deg_S(b) + 1) / 2$, where $\deg_S(b)$ denotes the number
of jobs that a machine $b$ receives in the semi-matching $S$, corresponds to minimizing the \textit{total completion time} of the jobs 
(optimizing with respect to this cost function automatically minimizes the maximal degree). 

It is well known that matchings are of maximal size if they do not admit \textit{augmenting paths} \cite{b57}. Augmenting
paths for matchings correspond to \textit{degree-minimizing paths} for semi-matchings. 
They first appeared in \cite{hllt03} under the name of \textit{cost-reducing-paths}, and they were used for the computation
of a semi-matching that minimizes a certain cost function. We use the term `degree-minimizing-path' since it is more appropriate 
in our setting. A degree-minimizing
path starts at a $B$ node of high degree, then alternates between edges of the semi-matching and edges outside
the semi-matching, and ends at another $B$ node of degree at least by two smaller than the degree of the 
starting point of the path. Flipping the
semi-matching and non-semi-matching edges of the path then generates a new semi-matching such that the large degree
of the start node of the path is decreased by $1$, and the small degree of the end node of the path is
increased by $1$. An \textit{optimal semi-matching} is defined in \cite{hllt03} to be one that does not admit any 
degree-minimizing paths.
It was shown in \cite{hllt03} that such a semi-matching is also optimal with respect to a large set of cost functions, 
including the minimization of the maximal degree as well as the minimization of the total completion time.
At present, the best existing algorithm for computing an optimal semi-matching \cite{fln2010} runs in time $\Order(\sqrt{|V|} |E| \log |V|)$ where
$V = A \cup B$. See \cite{fln2010} for a broader overview about previous work on semi-matchings (including works from the scheduling literature).

In this paper, we study \textit{approximation algorithms} for the semi-matching problem in different computational settings. 
The notion of approximation that we consider is with respect to the maximal degree: given a bipartite graph $G=(A, B, E)$ with
$n = |A|$, we are interested in computing a semi-matching $S$ such that $\degmax S \le c \cdot  \degmax S^*$, where $S^*$ denotes 
an optimal semi-matching, $\degmax$ denotes the maximal degree of a vertex w.r.t. a set of edges, and $c$ is the approximation factor. 
This notion of approximation corresponds to approximating the makespan
when the semi-matching is seen as a scheduling problem. This setting was already studied in e.g. \cite{anr95}.
%In prior works \cite{hllt2006,fln2010}, an optimal semi-matching is defined as one that minimizes a class of convex functions on the
%degrees of the $B$ nodes.

\vspace{0.2cm}

\textbf{Streaming Algorithms and Communication Complexity.}
\textit{Streaming Algorithms} fall into the category of massive data set algorithms. In many applications, the data that an algorithm
is called upon to process is too large to fit into the computer's memory. In order to cope with this problem, a streaming algorithm sequentially scans 
the input while using a random access memory of size sublinear in the length of the input stream. Multiple passes often help to further decrease the size of the
random access memory. \textit{Graph streams} are widely studied in the streaming model, %\cite{fkmsz04,hhls10}
and in the last years matching problems 
have received particular attention \cite{ag11,gkk12,kmm12,k13}. A graph stream is a sequence of the edges of the input graph with a priori
no assumption on the order of the edges. Particular arrival orders of the edges are studied in the literature and allow the design of algorithms
that depend on that order. Besides uniform random order \cite{kmm12}, the \textit{vertex arrival order} \cite{gkk12,k13} of edges of a bipartite graph is 
studied where edges incident to the same $A$ node arrive in blocks. 
Deciding basic graph properties such as connectivity already requires $\Omega(|V|)$ space \cite{fkmsz05}, where $V$ denotes the vertex set of a graph.
Many works considering graph streams allow an algorithm to use $\Order(|V| \polylog |V|)$ space. This setting is usually called the \textit{semi-streaming} setting.

Space lower bounds for streaming algorithms are often obtained via \textit{Communication Complexity}. There is an inherent link between
streaming algorithms and one-way $k$-party communication protocols. A streaming algorithm for a problem $P$ with space $s$ also serves as
a one-way $k$-party communication protocol for $P$ with communication cost $\Order(s k)$. % if the input $I = I[1] \dots I[n]$ of the 
%streaming algorithm is split to the $k$ parties such that party $i$ has $I[j_i] \dots I[j_{i+1}-1]$, for $j_1 = 1, j_{k+1} = n+1$ and $j_1 < j_2 < \dots < j_{k+1}$. 
Conversely, a lower bound on the size of any message of such a protocol is also a lower bound on the space requirements of a 
streaming algorithm. %Approaching the streaming complexity of problems from the communication complexity point of view 
%hence provides valuable insight.
Determining the communication complexity of problems is in itself an important task, however, the previously discussed link to
streaming algorithms provides an additional motivation.

\vspace{0.2cm}

\textbf{Our Contributions.} 
We initiate the study of the semi-matching problem in the streaming and the communication settings.
We present a deterministic one-pass streaming algorithm that for any $0 \le \epsilon \le 1$ uses space 
$\OrderT(n^{1+\epsilon})$ and computes an $\Order(n^{(1-\epsilon)/2})$ approximation to the semi-matching problem 
(\textbf{Theorem~\ref{theorem:one-pass-streaming}})\footnote{We write $\OrderT(n)$ to denote $\Order(n \polylog n)$.}.
Furthermore, we show that with $\Order(\log n)$ passes we can compute an $\Order(\log n)$
approximation with space $\OrderT(n)$ (\textbf{Theorem~\ref{theorem:log-n-approx}}).
%Our streaming algorithms do not assume any particular arrival order of the edges of the input graph. 

In the two-party one-way communication setting, %Alice holds a graph $G_1 = (A, B, E_1)$ and Bob holds a graph $G_2 = (A, B, E_2)$.
%Alice sends a message $M$ to Bob and Bob outputs a semi-matching in the graph $G_1 \cup G_2$. 
we show that for any $\epsilon > 0$,
deterministic communication protocols that compute an $\Order(n^{\frac{1}{(1+\epsilon)c + 1}})$ approximation to the semi-matching problem
require a message of size at least $cn$ bits (\textbf{Theorem~\ref{thm:communicatio-lb}}). We present two deterministic protocols communicating 
$n$ and $2n$ edges that compute an $\Order(\sqrt{n})$ approximation and an $\Order(n^{1/3})$ approximation, respectively (\textbf{Theorem~\ref{thm:communication-ub}}).
%The construction of our protocols are based on the study of

%Furthermore, we present results concerning the structure of semi-matchings. 
While it was known that an optimal 
semi-matching contains a maximum matching \cite{hllt03}, we show that there is
a hierarchical decomposition of an optimal semi-matching into maximum matchings (\textbf{Lemma~\ref{lemma:no-deg-min-path}}). 
Similarly, we show that semi-matchings that do not admit length-two degree-minimizing paths can be decomposed into 
maximal matchings (\textbf{Lemma~\ref{lemma:no-length-2-deg-min-path}}). The latter result allows us to prove that
the maximal degree of a semi-matching that does not admit a length-two degree-minimizing path is at most $\lceil \log(n + 1) \rceil$
times the maximal degree of an optimal semi-matching (\textbf{Theorem~\ref{theorem:log-n-approximation}}). %This theorem is
%similar in spirit to the fact that maximal matchings are  of size at least $\frac{1}{2}$ times the size of a maximum matching.
%We observed that the $\lceil \log(n + 1) \rceil$ competitive\footnote{In \cite{anr95} they prove that the algorithm is
%$\lceil \log n \rceil + 1$ competitive. } online algorithm for computing semi-matchings in \cite{anr95} 
%computes a semi-matching that does not admit length-two degree-minimizing paths. 
%Hence, Theorem~\ref{theorem:log-n-approximation} immediately proves
%the $\lceil \log(n + 1) \rceil$ competitive ratio, while in \cite{anr95} this result is obtained by a complicated charging scheme.

%Due to space restrictions, 
%for an overview about existing works on semi-matchings (also from the scheduling literature) we refer the reader to \cite{fln2010}.

\vspace{0.2cm}

\textbf{A semi-streaming algorithm for vertex arrival order.}
In \cite{anr95}, the semi-matching problem is studied in the online model (seen as a scheduling problem). 
In this model, the $A$ vertices arrive online together with their incident edges, and it has to be irrevocably decided to which $B$ node 
an $A$ node is matched. It is shown that the greedy algorithm 
matching an $A$ node to the $B$ node that currently has the smallest degree is $\lceil \log(n + 1) \rceil$ competitive,
and that this result is tight. %In Section~\ref{section:structure}, we present an alternative proof for the $\log(n) + 1$ competitive 
%ratio of this online algorithm. 
This algorithm can also be seen as a one-pass $\lceil \log(n + 1) \rceil$ approximation semi-streaming 
algorithm (meaning $\OrderT(n)$ space) for the semi-matching problem when the input stream is in vertex arrival order. 
Note that our one-pass algorithm does not assume any order on the input sequence, and when allowing $\OrderT(n)$ space it 
achieves an $\Order(\sqrt{n})$-approximation.
%Note that our algorithms do not make any assumptions on the input order. We mention that when allowing $\OrderT(n)$ space, 
%we obtain a $1$-pass $\Order(\sqrt{n})$ approximation algorithm (Theorem~\ref{theorem:one-pass-streaming}) and a $\log(n)$-pass 
%$4\log(n)$ approximation algorithm (Theorem~\ref{theorem:log-n-approx}).
%Furthermore, the fact that an optimal maximum matching contains always a maximum matching is proved in \cite{hllt2006}.
%They were also the first to use the notion of degree-minimizing-paths. However, since 
%they see the problem of finding an optimal semi-matchings as a problem of minimizing an objective function, they denote
%such paths as \textit{cost-reducing-paths}. We consider our terminology more appropriate in our setting.
%Last, we mention that the work of \cite{gkk12} is similar in spirit to our paper. They studied the streaming complexity 
%and the communication complexity of bipartite maximum matching. 

\vspace{0.2cm}

\textbf{Techniques.} 
Our streaming algorithms are based on the following greedy algorithm. 
%that was already used in \cite{kmm12}. 
Fix a maximal degree 
$d_{\max}$ (for instance $d_{\max} = n^{1/4}$) and greedily add edges to a set $S_1$ such that the maximal degree of 
a $B$ node in $S_1$ does not exceed $d_{\max}$, and the degree of any $A$ node in $S_1$ is at most $1$. %We call such a structure an
%\textit{incomplete $d_{\max}$-bounded semi-matching}. 
This algorithm leaves at most $\Order(n/d_{\max})$ $A$ vertices unmatched in $S_1$. % since for each unmatched vertex $a$, the mate $b$ of $a$ in the perfect matching must have $d_{\max}$ other $A$ vertices incident 
%in $S_1$ since otherwise $a$ would have been matched to $b$ in $S_1$).
To match the yet unmatched vertices, we use a second greedy algorithm that we run in parallel to the first one. We fix a parameter 
$d'$ appropriately (if $d_{\max} = n^{1/4}$ then we set $d' = n^{1/2}$) and for all vertices $a \in A$ we store arbitrary 
$d'$ edges incident to $a$ in a set $E'$. 
Then, we compute an optimal semi-matching $S_2$ of the unmatched vertices in $S_1$ and the $B$ nodes only considering the edges in $E'$. 
We prove that such a semi-matching has bounded maximal degree (if $d_{\max} = n^{1/4}$ and $d' = n^{1/2}$ then this degree is $n^{1/4}$). 
The set $S_1 \cup S_2$ is hence a semi-matching of maximal degree $d_{\max} + \degmax S_2$ % (if $d_{\max} = n^{1/4}$ and $d' = n^{1/2}$ then this is $2n^{1/4}$)
and the space requirement of this algorithm is $\OrderT(n d')$. % (if $d' = n^{1/2}$ then this $\OrderT(n^{1+1/2})$).
In Section~\ref{section:streaming} we generalize this idea for any $0 \le \epsilon \le 1$ to obtain one-pass algorithms with
approximation factors $\Order(n^{1/2(1-\epsilon)})$ using space $\OrderT(n^{1+\epsilon})$,
and a $\log(n)$-pass algorithm with approximation factor $\Order(\log n)$ using space $\OrderT(n)$.

In the two-party one-way communication setting, the edge set $E$ of a bipartite graph $G=(A, B, E)$ is split among two players,
Alice and Bob. %holds a subgraph $G_1 = (A, B, E_1) \subseteq G$ and Bob holds a subgraph 
%$G_2 = (A, B, E_2) \subseteq G$. 
Alice sends a message to Bob and Bob outputs a semi-matching of $G$.
Our communication upper bounds make use of what we call a \textit{$c$-semi-matching skeleton} (or simply $c$-skeleton). 
A $c$-skeleton of a bipartite graph $G = (A, B, E)$ is a subset of edges $S \subseteq E$ such that for any 
$A' \subseteq A: \degmax \semi(A', B, S) \le c \cdot \degmax \semi(A', B, E)$
where $\semi(A', B, E')$ denotes an optimal semi-matching between $A'$ and $B$ using edges in $E'$. 
We show that if Alice sends a $c$-skeleton $S$ of her subgraph to Bob, and Bob computes an optimal 
semi-matching using his edges and the skeleton, then the resulting semi-matching is a $c+1$ approximation. 
We show that there is an $\Order(\sqrt{n})$-skeleton consisting of $n$ edges, and that there is an 
$\Order(n^{1/3})$-skeleton consisting of $2n$ edges. 
It turns out that an optimal semi-matching is an $\Order(\sqrt{n})$-skeleton, and we show how an $\Order(n^{1/2})$-skeleton
can be improved to an $\Order(n^{1/3})$-skeleton by adding additional $n$ edges.
%To obtain an $\Order(n^{1/3})$-skeleton, we start with the an $\Order(n^{1/2})$-skeleton, and we compute for each subset $A'$ that is matched in
%the $\Order(n^{1/2})$-skeleton to the same $B$ node a new semi-matching. The details of this construction are provided in 
%Section~\ref{section:comm-upper-bound}. 
These skeletons are almost optimal: 
we show that for any $\epsilon > 0$, an $\Order(n^{\frac{1}{(1+\epsilon)c+1}})$-skeleton has at least $cn$ edges.
Inspired by the prior lower bound, we prove that for any $\epsilon > 0$, the deterministic one-way two-party communication complexity 
of approximating semi-matchings within a factor $\Order(n^{\frac{1}{(1+\epsilon)c+1}})$ is at least $cn$ bits.

%Hard instances for these lower bounds are complete bipartite graphs with a carefully chosen number of $B$ vertices.

%This strategy is essentially optimal: the lower bounds
%on the size of a $c$-skeleton coincide almost with the lower bounds for any deterministic semi-matching protocol. 

In order to prove our structure lemmas on semi-matchings, we make use of degree-minimizing paths. Our results on the decomposition
of semi-matchings into maximum and maximal matchings directly relate the absence of degree-minimizing paths to the absence of
augmenting paths in matchings. See Section~\ref{section:structure} for details. %Degree-minimizing paths seem to be the right combinatorial object 
%to reason about the quality of semi-matchings, see Section~\ref{section:structure} for details.

%Consider a semi-matching $S$ that does not admit a length-two degree-minimizing path. For each $b \in B$, label then the edges incident to it
%aribitrarily by $1$ to $\deg(b)$. Then the set of edges labeled by $1$ forms a maximal matching: if this was not the case then there
%would be a degree-minimizing path of length $2$ in the semi-matching. This idea can then be applied to edges labeled by some $i > 1$. 
%Edges labeled by $i$ form a maximal matching in the induced subgraph $G|_{A_i \times B_i}$, where $A_i$ is the set of $A$ vertices that 
%have incident edges labeled $\ge i$, and $B_i$ is the set of vertices that have an incident edge labeled by $i-1$, for details see Lemma~\ref{lemma:no-length-2-deg-min-path}. 
%This identified structure allows us then to prove via an intermediate lemma, Lemma~\ref{lemma:matching-half-of-the-nodes}, 
%that semi-matchings that do not admit length-two degree-minimizing paths have a maximal degree of at most $\log(n) + 1$ times the maximal 
%degree of an optimal semi-matching, see Theorem~\ref{theorem:log-n-approximation}. 

\vspace{0.2cm}

\textbf{Organization.} After presenting notations and definitions in Section~\ref{section:preliminaries},
we present our streaming algorithms in Section~\ref{section:streaming}. We then discuss the one-way two-party communication setting in
Section~\ref{section:comm}. We conclude with Section~\ref{section:structure}, where we present our results on the structure of semi-matchings.

\section{Notations and Definitions} \label{section:preliminaries}
Let $G = (A, B, E)$ be a bipartite graph and let $n = |A|$.
For ease of presentation, we assume that $|B|$ is upper-bounded by a polynomial in $n$.  
%Let $E' \subseteq E, A' \subseteq A$, and $B' \subseteq B$. 
%Furthermore, let $v \in A \cup B, a \in A$, and $b \in B$. 
Let $e \in E$ be an edge connecting nodes $a \in A$ and $b \in B$. Then, we 
write $A(e)$ to denote the vertex $a$, $B(e)$ to denote the vertex $b$, and $ab$ to denote $e$. Furthermore, for a subset $E' \subseteq E$, we define
$A(E') = \bigcup_{e\in E'} A(e)$ (respectively $B(E')$). For subsets $A' \subseteq A$ and $B' \subseteq B$ 
we write $E'|_{A' \times B'}$ to denote the subset of edges of $E'$ whose endpoints are all in $A' \cup B'$. We 
denote by $E'(a)$ the set of edges of $E' \subseteq E$ that have an endpoint in vertex $a$, and $E'(A')$ the set of edges that
have endpoints in vertices of $A'$, where $A' \subseteq A$ (similarly we define $E'(B')$ for $B' \subseteq B$).

For a node $v \in A \cup B$, the \textit{neighborhood} of $v$ is the set of nodes that are connected to $v$ and we denote it 
by $\Gamma(v)$. For a subset $E' \subseteq E$, we write $\Gamma_{E'}(v)$ to denote the neighborhood of $v$
in the graph induced by $E'$. Note that by this definition $\Gamma(v) = \Gamma_E(v)$.
For a subset $E' \subseteq E$, we denote by $\deg_{E'}(v)$ the \textit{degree} in $E'$ of a node $v \in V$, 
which is the number of edges of $E'$ with an endpoint in $v$. We define $\degmax E' := \max_{v \in A \cup B} \deg_{E'}(v)$. 

\vspace{0.1cm}

\textbf{Matchings.} A \textit{matching} is a subset $M \subseteq E$ such that $\forall v \in A \cup B: \deg_M(v) \le 1$.
A \textit{maximal matching} is a matching that is inclusion-wise maximal, i.e. it can not be enlarged by 
adding another edge of $E$ to it.
A \textit{maximum matching} is a matching of maximal size. % and we denote it by $M^*$. 
A \textit{length $p$ augmenting path} ($p \ge 3$, $p$ odd) wrt. a matching $M$ is a path $P = (v_1, \dots, v_{p+1})$ such that
$v_1, v_{p+1} \notin A(M) \cup B(M)$ and for $i \le 1/2(p-1): v_{2i}v_{2i+1} \in M$, and $v_{2i-1}v_{2i} \notin M$. 

\vspace{0.1cm}

\textbf{Semi-Matchings.} A \textit{semi-matching} of $G$ is a subset $S \subseteq E$ such that $\forall a \in A: \deg_S(a) = 1$.
A \textit{degree-minimizing path} $P = (b_1, a_1, \dots, b_{k-1}, a_{k-1}, b_k)$ with respect to a
semi-matching $S$ is a path of length $2k$ ($k \ge 1$) such that for all $i \le k:$ $(a_i, b_i) \in S$,
for all $i \le k-1: (a_i, b_{i+1}) \notin S$, and 
$\deg_S(b_1) > \deg_S(b_2) \ge \deg_S(b_3) \ge \dots \ge \deg(b_{k-1}) > \deg(b_k)$.
An \textit{optimal semi-matching} $S^* \subseteq E$ is a semi-matching that does not admit any degree-minimizing-paths.
For subsets $A' \subseteq A, B' \subseteq B, E' \subseteq E$, we denote by $\semi(A', B', E')$ an optimal semi-matching
in the graph $G' = (A', B', E')$, and we denote by $\semi_2(A', B', E')$ a semi-matching
that does not admit degree-minimizing paths of length $2$ in $G'$.

\vspace{0.1cm}

\textbf{Incomplete $d$-bounded Semi-Matchings.} Let $d$ be an integer. Then an \textit{incomplete $d$-bounded semi-matching} of $G$ is a subset 
$S \subseteq E$ such that $\forall a \in A: \deg_S(a) \le 1$ and $\forall b \in B: \deg_S(b) \le d$. 
For subsets $A' \subseteq A, B' \subseteq B, E' \subseteq E$, we write $\isemi_{d}(A', B', E')$ to denote an incomplete 
$d$-bounded semi-matching of maximal size in the graph $G' = (A', B', E')$. 

\vspace{0.1cm}

\textbf{Approximation.} We say that an algorithm (or communication protocol) is a $c$-approximation algorithm (resp. communication protocol) 
to the semi-matching problem if it outputs a semi-matching $S$ such that $\degmax S \le c \cdot \degmax S^*$, where $S^*$ denotes an optimal semi-matching. We note that this measure was previously 
used for approximating semi-matching, e.g, in \cite{anr95}.

\section{Streaming Algorithms} \label{section:streaming}

To present our streaming algorithms, we
%In this section, we 
describe an algorithm, $\ALGOsemi(G, s, d, p)$ (Algorithm~\ref{algo:streaming}), that computes an 
incomplete $2dp$-bounded semi-matching in the 
graph $G$ using space $\OrderT(s)$, and makes at most $p \ge 1$ passes over the input stream. If appropriate parameters are 
chosen, then the output is not only an incomplete semi-matching, but also a semi-matching. 
We run multiple copies of this algorithm with different parameters in parallel
%We then choose appropriate parameters and run a number of versions of this algorithm in parallel 
in order to obtain a one-pass algorithm for the semi-matching problem (Theorem~\ref{theorem:one-pass-streaming}).
% we show how to apply appropriate parameters in order to obtain a 
Using other parameters, we also obtain 
%we show how to apply appropriate parameters to obtain 
a $\log(n)$-pass algorithm, as stated in Theorem~\ref{theorem:log-n-approx}.  %for the semi-matching problem. 

%\vspace{-0.5cm}
\begin{algorithm}[ht]
 \caption{Skeleton for approximating semi-matchings: $\ALGOsemi(G, s, d, p)$ \label{algo:streaming}}
 \begin{algorithmic}
 \REQUIRE $G = (A, B, E)$ is a bipartite graph
 \STATE $S \gets \varnothing$
 \STATE \textbf{repeat} at most $p$ times or until $|A(S)| = |A|$
  \STATE $\quad S \gets S \cup \ALGOincomplete(G|_{(A \setminus A(S)) \times B}, s, d)$
 \STATE \textbf{end repeat}
 \RETURN $S$
 \end{algorithmic}
\end{algorithm} 

%\vspace{-1.1cm}

%\vspace{-1.1cm}

\begin{algorithm}[ht]
\caption{Computing incomplete semi-matchings: $\ALGOincomplete(G, s, d)$ \label{algo:streaming-subroutine}}
 \begin{algorithmic}
 \REQUIRE $G = (A, B, E)$ is a bipartite graph
 \STATE $k \gets s/|A|$, $S_1 \gets \varnothing$, $E' \gets \varnothing$
 \WHILE{$\exists$ an edge $ab$ in stream}
 \STATE \textbf{if} $ab \notin A \times B$ \textbf{then} \textbf{continue} 
 \STATE \textbf{if} $\deg_{S_1}(a) = 0$ and $\deg_{S_1}(b) < d$ \textbf{then} $S_1 \gets S_1 \cup \{ ab \}$ 
%  \IF{$\deg_{S_1}(a) = 0$ and $\deg_{S_1}(b) < d$} 
%   \STATE $S_1 \gets S_1 \cup \{ ab \}$
%  \ENDIF
%  \IF{$\deg_{E'}(a) < k$}
 \STATE \textbf{if} $\deg_{E'}(a) < k$ \textbf{then} $E' \gets E' \cup \{ ab \}$ 
%   \STATE $E' \gets E' \cup \{ ab \}$
%  \ENDIF
 \ENDWHILE
% \STATE $A' \gets A \setminus A(S_1)$
 \STATE $S_2 \gets \isemi_d(E'|_{(A \setminus A(S_1)) \times B})$
 \STATE $S \gets S_1 \cup S_2$
 \RETURN $S$
 \end{algorithmic}
\end{algorithm}

%\vspace{-0.5cm}

$\ALGOsemi(G, s, d, p)$ starts with an empty incomplete semi-matching $S$ and adds edges to $S$ by invoking 
$\ALGOincomplete(G, s, d)$ (Algorithm~\ref{algo:streaming-subroutine}) on the subgraph of the 
as yet unmatched $A$ vertices in $S$ and all $B$ vertices. Each invocation of $\ALGOincomplete(G, s, d)$ makes one pass over the input stream and returns
a $2d$-bounded incomplete semi-matching while using space $\OrderT(s)$. Since we make at most $p$ passes, the resulting incomplete
semi-matching has a maximal degree of at most $2dp$.

$\ALGOincomplete(G, s, d)$ collects edges greedily from graph $G$ and puts them into an incomplete $d$-bounded semi-matching $S_1$ and a set $E'$. 
An edge $e$ from the input stream is put into $S_1$ if $S_1 \cup \{e \}$ is still an incomplete $d$-bounded semi-matching. 
An edge $e=ab$ is added to $E'$ if the degree of $a$ in $E' \cup \{e \}$
is less or equal to a parameter $k$ which is chosen to be $s/|A|$ in order to ensure that the algorithm does not exceed space $\OrderT(s)$.
The algorithm returns an incomplete $2d$-bounded semi-matching that consists of $S_1$ and $S_2$, where
$S_2$ is an optimal incomplete $d$-bounded semi-matching between the $A$ vertices that are not matched in $S_1$ and all $B$ vertices, 
using only edges in $E'$. 

We lower-bound the size of $S_2$ in Lemma~\ref{lemma:d-reg-graph-semi}. We prove that
for any bipartite graph $G= (A, B, E)$ and any $k > 0$, if we store for each $a \in A$ any $\max \{k, \deg_G(a) \}$ 
incident edges to $a$, then we can compute an incomplete $d$-bounded semi-matching of size at least $\min \{k d, |A| \}$
using only those edges, where $d$ is an upper-bound on the maximal degree of an optimal semi-matching between $A$ and $B$ in $G$.

Lemma~\ref{lemma:d-reg-graph-semi} is then used in
the proof of Lemma~\ref{lemma:isemi}, where we show a lower bound on the size of the output $S_1 \cup S_2$ of $\ALGOincomplete(G, s, d)$.

\begin{lemma} \label{lemma:d-reg-graph-semi}
 Let $G = (A, B, E)$ be a bipartite graph, let $k > 0$ and let $d \ge \degmax \semi(A, B, E)$. 
Furthermore, let $E' \subseteq E$ be a subset of edges such that for all $a \in A: \deg_{E'}(a) = \min \{ k, \deg_E(a) \}$. Then there is an incomplete 
$d$-bounded semi-matching $S \subseteq E'$ such that $|S| \ge \min \{k d, |A| \}$.
\end{lemma}

\begin{proof}
 Let $d^* = \degmax \semi(A, B, E)$.
 We explicitly construct an incomplete semi-matching $S$. Let $A_0 \subseteq A$ such that for all $a \in A_0: \deg_{E'}(a) = \deg_E(a)$, 
and let $A_1 = A \setminus A_0$. Let $S_0 = \semi(A_0, B, E)$. Clearly, $\degmax S_0 \le d^*$. We construct now $S$ as follows. 

Start with $S = S_0$, and then add greedily edges in any order from $E'|_{A_1 \times B}$ to $S$ such that $S$ remains an incomplete semi-matching 
with maximal degree $d$. Stop as soon as there is no further edge that can be added to $S$. 

We prove that $S$ contains at least 
$\min \{k d, |A| \}$ edges. To see this, either all nodes of $A$ are matched in $S$, or there is at least one node $ a \in A_1$ that is not matched in $S$ 
(note that all nodes in $A_0$ are matched in $S$). Since $\deg_{E'}(a) = k$, all nodes $b \in \Gamma_{E'}(a)$ have degree $d$ since otherwise $a$ 
would have been added to $S$. This implies that there are at least $k \cdot d$ nodes matched in $S$ which proves the lemma. \qed
\end{proof}

\begin{lemma} \label{lemma:isemi}
 Let $G = (A, B, E)$ be a bipartite graph, let $s \ge |A|$ and let $d \ge \degmax \semi(A, B, E)$. 
Then $\ALGOincomplete(G, s, d)$ (see Algorithm~\ref{algo:streaming-subroutine}) uses $\OrderT(s)$ space and outputs an incomplete $2d$-bounded 
semi-matching $S$ such that $|S| \ge \min \{|A| \frac{d}{d+d^*} + \frac{ds}{|A|}, |A| \}$.

\end{lemma}

\begin{proof}
 The proof refers to the variables of Algorithm~\ref{algo:streaming-subroutine} and the values they take at the end of the algorithm. 
Furthermore, let $S^* = \semi(A, B, E)$, $d^* = \degmax S^*$, and let $A' = A \setminus A(S_1)$.

Firstly, we lower-bound $|S_1|$. Let $a \in A'$ and $b = S^*(a)$. 
Then $\deg_{S_1}(b) = d$ since otherwise $a$ would have been matched in $S_1$. Hence, we obtain 
$|A(S_1)| \ge d |B(S^*(A'))| \ge d |A'|/d^*$, where
the second inequality holds since the maximal degree in $S^*$ is $d^*$. Furthermore, since $A' = A \setminus A(S_1)$ and 
$|S_1| = |A(S_1)|$, we obtain $|S_1| \ge |A| \frac{d}{d+d^*}$.
We apply Lemma~\ref{lemma:d-reg-graph-semi} on the graph induced by the edge set $E'|_{A' \times B}$. We obtain that 
$|S_2| \ge \min \{d s/|A|, |A'| \}$ and consequently $|S| = |S_1| + |S_2| \ge \min \{|A| \frac{d}{d+d^*} + \frac{ds}{|A|}, |A| \}$.

Concerning space, the dominating factor is the storage space for the at most $k+1$ edges per $A$ vertex, and
hence space is bounded by $\OrderT(k |A|) = \OrderT(s)$. \qed
\end{proof}

In the proof 
of Theorem~\ref{theorem:one-pass-streaming}, for $0 \le \epsilon \le 1$ we show that
$\ALGOsemi(G$, $n^{1+\epsilon}$, $n^{1/2(1-\epsilon)} d'$, $1)$ returns a semi-matching if
$d'$ is at least the maximal degree of an optimal semi-matching.
Using a standard technique, we run $\log(n)+1$ copies of $\ALGOsemi$
for all $d' = 2^i$ with $0 \le i \le \log(n)$
and we return the best semi-matching, obtaining a $1$-pass algorithm.
%$d' \ge d^*$, where $d^*$ is the maximal degree of an optimal semi-matching. Since $d^*$
%is unknown in advance, we use a standard technique to obtain a value of $d'$ such that $d^* \le d' < 2d^*$. 
%We run $\log(n)+1$ copies of $\ALGOsemi$ for all $d' = 2^i$ with $0 \le i \le \log(n)$,
%and we return the best semi-matching.
We use the same idea in Theorem~\ref{theorem:log-n-approx}, where we obtain a $4\log(n)$ approximation algorithm
that makes $\log(n)$ passes and uses space $\OrderT(n)$.
%By chosing appropriate parameters for Algorithm~\ref{algo:streaming}, we obtain the following 
%two theorems (the proofs can be found in the Appendix).

\begin{theorem} \label{theorem:one-pass-streaming}
 Let $G = (A, B, E)$ be a bipartite graph with $n = |A|$. For any $0 \le \epsilon \le 1$ there is a one-pass streaming algorithm 
using $\OrderT(n^{1+\epsilon})$ space that computes a $4 n^{1/2(1-\epsilon)}$ approximation to the semi-matching problem.
\end{theorem}

\begin{proof}
  We run $\log(n) + 1$ copies of Algorithm~\ref{algo:streaming} in parallel as follows. 
For $0 \le i \le \lceil \log(n) \rceil$ let $S_i = \ALGOsemi(G, n^{1+\epsilon}, n^{1/2(1-\epsilon)} 2^i, 1)$ and choose among the $S_i$
a semi-matching $S_k$ such that $|S_k| = n$ and for any other $S_l$ with $|S_l| = n: \degmax S_k \le \degmax S_l$.

We show now that there is a $S_j$ which is a semi-matching that fulfills the desired approximation guarantee. Let $S^* = \semi(A, B, E)$ and $d^* = \degmax(S^*)$. 
Then define $j$ to be such that $d^* \le 2^j < 2d^*$ and let $d = n^{1/2(1-\epsilon)} 2^j$. $S_j$ is the output of a call to $\ALGOincomplete(G, n^{1+\epsilon}, d)$. By Lemma~\ref{lemma:isemi}, $S_j$ is of size at least $\min \{n \frac{d}{d+d^*} + dn^{\epsilon}, |A| \}$ which equals $|A|$ for our choice of $d$. This proves that all $a \in A$ are matched in $S_j$. By Lemma~\ref{lemma:isemi}, $\degmax S_j \le 2 d$ which is less or equal to $4 n^{1/2(1-\epsilon)} d^*$. Hence, $S_j$ is a $4 n^{1/2(1-\epsilon)}$ approximation. 

The space requirement is $\log n$ times the space requirement for the computation of a single $S_i$ which is dominated by the space 
requirements of Algorithm~\ref{algo:streaming-subroutine}. By Lemma~\ref{lemma:isemi}, this is $\OrderT(n^{1+\epsilon})$, and hence the 
algorithm requires $\OrderT(n^{1+\epsilon} \log n) = \OrderT(n^{1+\epsilon})$ space. \qed
\end{proof}

\begin{theorem} \label{theorem:log-n-approx}
Let $G = (A, B, E)$ be a bipartite graph with $n = |A|$. There is a $\log(n)$-pass streaming algorithm using space $\OrderT(n)$ that computes a $4 \log(n)$ approximation to the semi-matching problem.
\end{theorem}

\begin{proof}
 As in the proof of Theorem~\ref{theorem:one-pass-streaming}, we run $\log(n) + 1$ copies of Algorithm~\ref{algo:streaming} in parallel.
For $0 \le i \le \lceil \log(n) \rceil$ let $S_i = \ALGOsemi(G, n, 2^i, \log(n))$ and choose among the $S_i$
a semi-matching $S_k$ such that $|S_k| = n$ and for any other $S_l$ with $|S_l| = n: \degmax S_k \le \degmax S_l$.

We show now that there is a $S_j$ which is a semi-matching that fulfills the desired approximation guarantee. Let $S^* = \semi(A, B, E)$ and $d^* = \degmax(S^*)$. Then define $j$ to be such that $d^* \le 2^j < 2d^*$ and let $d = 2^j$. 
$S_j$ is the output of a call to $\ALGOsemi(G, n, d, \log(n))$. In each iteration, the algorithm calls 
$\ALGOincomplete(G', n, d)$, where $G'$ is the subgraph of $G$ of the not yet matched $A$ vertices and the $B$ vertices.
By Lemma~\ref{lemma:isemi}, at least a $\frac{d}{d+d^*} \ge 1/2$ fraction of the unmatched $A$ vertices is matched since
$d \ge d^*$, and the maximal degree of the incomplete semi-matching returned by $\ALGOincomplete(G', n, d)$ is at most $2d$. 
Hence, after $\log(n)$ iterations, all $A$ vertices are matched. Since $d < 2d^*$ and the algorithm performs at 
most $\log(n)$ iterations, the algorithm returns a $4 \log(n)$ approximation.  

Each copy of Algorithm~\ref{algo:streaming} uses space $\OrderT(n)$ and since we run $\Order(\log n)$ the required
space is $\OrderT(n)$. \qed
\end{proof}

\section{Two-party Communication Complexity} \label{section:comm}

We now consider one-way two-party protocols which are given a bipartite graph $G=(A, B,E)$ as input,
%We assume that the sets $A$ and $B$ are known in advance to both Alice and Bob %(this only strengthens our lower bound) 
such that  $E_1 \subseteq E$ is given to Alice and  $E_2 \subseteq E$ is given to Bob. 
%We do not impose any particular format on the output of the protocol. 
Alice sends a single message to Bob, and Bob outputs
a valid semi-matching $S$ for $G$.
A central idea for our upper and lower bounds is what 
we call a \textit{$c$-semi-matching skeleton} (or \textit{$c$-skeleton}). Given a bipartite graph $G=(A, B, E)$,
we define a $c$-semi-matching skeleton to be a subset of edges $S \subseteq E$ such that 
$\forall A' \subseteq A: \degmax \semi(A', B, S) \le c \cdot  \degmax \semi(A', B, E)$.
%where $c$ is a function of $n = |A|$. 
We show how to construct an $\Order(\sqrt{n})$-skeleton of size $n$, and an $\Order(n^{1/3})$-skeleton of
size $2n$. We show that if Alice sends a $c$-skeleton of her subgraph $G=(A, B, E_1)$ to Bob, then Bob can 
output a $c+1$-approximation to the semi-matching problem. Using our skeletons, we thus obtain 
one-way two party communication protocols for the semi-matching problem with approximation factors
$\Order(\sqrt{n})$ and $\Order(n^{1/3})$, respectively (Theorem~\ref{thm:communication-ub}).
Then we show that for any $\epsilon > 0$, an $\Order(n^{\frac{1}{(1+\epsilon)c + 1}})$-skeleton 
requires at least $cn$ edges. This renders our $\Order(\sqrt{n})$-skeleton and our $\Order(n^{1/3})$-skeleton tight up to a constant.

%Finally, we show in Subsection~\ref{section:comm-lb} that the numbers of edges required for an $\Order(c)$-skeleton coincides with 
%the deterministic one-way two-party communication complexity for obtaining an $\Order(c)$ approximation in bits. 
\subsection{Upper Bound} 
Firstly, we discuss the construction of two skeletons. In Lemma~\ref{lemma:upper-bound-sparsifier-1}, we show that
an optimal semi-matching is an $\Order(\sqrt{n})$-skeleton, and in Lemma~\ref{lemma:upper-bound-sparsifier-2}, we show
how to obtain a $\Order(n^{1/3})$-skeleton.
In these constructions, we use the following key observation: Given a bipartite
graph $G = (A, B, E)$, let $A' \subseteq A$ be such that $A'$ has minimal expansion, meaning that 
$A' = \argmin_{A'' \subseteq A} \frac{|\Gamma(A'')|}{|A''|}$. The maximal degree in a semi-matching is then clearly
at least $\lceil \frac{|A'|}{|\Gamma(A')|} \rceil$ since all vertices of $A'$ have to be matched to its neighborhood.
However, it is also true that the maximal degree of a semi-matching \textit{equals} $\lceil \frac{|A'|}{|\Gamma(A')|} \rceil$.
A similar fact was used in \cite{gkk12} for fractional matchings, and also in \cite{krt99}. 
For completeness, we are going to prove this fact in Lemma~\ref{lemma:exp-max-degree}. This proof requires 
the following technical lemma, Lemma~\ref{lemma:fractional-matching}.

\begin{lemma} \label{lemma:fractional-matching}
 Let $G = (A, B, E)$ be a bipartite graph and let $A' \subseteq A$ such that $|\Gamma(A')| \le |A'|$. Then:
 \begin{equation*} %\label{eqn:923}
\forall A'' \subseteq A': \frac{|\Gamma(A'')|}{|A''|} \ge \frac{|\Gamma(A')|}{|A'|} \, \Rightarrow \, \degmax \semi(A', B, E) \le \lceil \frac{|A'|}{|\Gamma(A')|} \rceil. 
 \end{equation*}
\end{lemma}

\begin{proof}
The proof is by contradiction. Let $d = \lceil \frac{|A'|}{|\Gamma(A')|} \rceil$, 
$S = \semi(A', B, E)$ and suppose that $\degmax S \ge d + 1$. 
We construct now a set $\tilde{A} \subset A'$ such that 
$\frac{|\Gamma(\tilde{A})|}{|\tilde{A}|} < \frac{|\Gamma(A')|}{|A'|}$ contradicting the premise of the lemma.

To this end, we define two sequences $(A_i)_i$ with $A_i \subseteq A'$ and $(B_i)_i$ with $B_i \subseteq \Gamma(A')$.
Let $b \in \Gamma(A')$ be a node with $\deg_S(b) \ge d + 1$ and let $B_1 = \{ b \}$. 
We define
\begin{eqnarray}
\nonumber A_i & = & \Gamma_S(B_i), \label{eqn:927} \\ % \setminus \cup_{j < i} B_j, \\
B_{i+1} & = & \Gamma(A_i) \setminus \cup_{j\le i} B_j.  % \setminus \cup_{j \le i} A_i . 
\end{eqnarray}
This setting is illustrated in Figure~\ref{figure:lemma-fractional-matching}. Note that all $A_i$ and all $B_i$ are disjoint. 
Let $k$ be such that $|A_k| > 0$ and $|A_{k+1}| = 0$. Then we set $\tilde{A} = \bigcup_{i=1}^k A_i$. 

By construction of the sequence $(B_i)_i$,
it is clear that for any $b' \in \cup B_i: \deg_S(b') \ge \deg_S(b) - 1$, since otherwise there is a 
degree-minimizing path from $b$ to $b'$ contradicting the definition of $S$. 
Then, by Equation~\ref{eqn:927}, we obtain for all $i$ that 
$|A_i| \ge |B_i| (\deg_S(b) - 1)$ which implies that $|A_i| \ge d |B_i|$ since $\deg_S(b) \ge d + 1$. 
Remind that $|A_1| \ge d + 1$. We compute

\begin{equation*}
 \frac{|\Gamma(\tilde{A})|}{|\tilde{A}|} = \frac{|B_1| + \sum_{2 \le i \le k}|B_i|}{|A_1| + \sum_{2 \le i \le k}|A_i|} \le 
\frac{1+\sum_{2 \le i \le k} |B_i|}{(d + 1) + \sum_{2 \le i \le k}|B_i|d} < \frac{1}{d} \le \frac{|\Gamma(A')|}{|A'|},
\end{equation*}
and we obtain a contradiction to the premise of the lemma. \qed
\end{proof}

\begin{figure}[ht!]
\begin{center}
 \includegraphics[height=4.5cm]{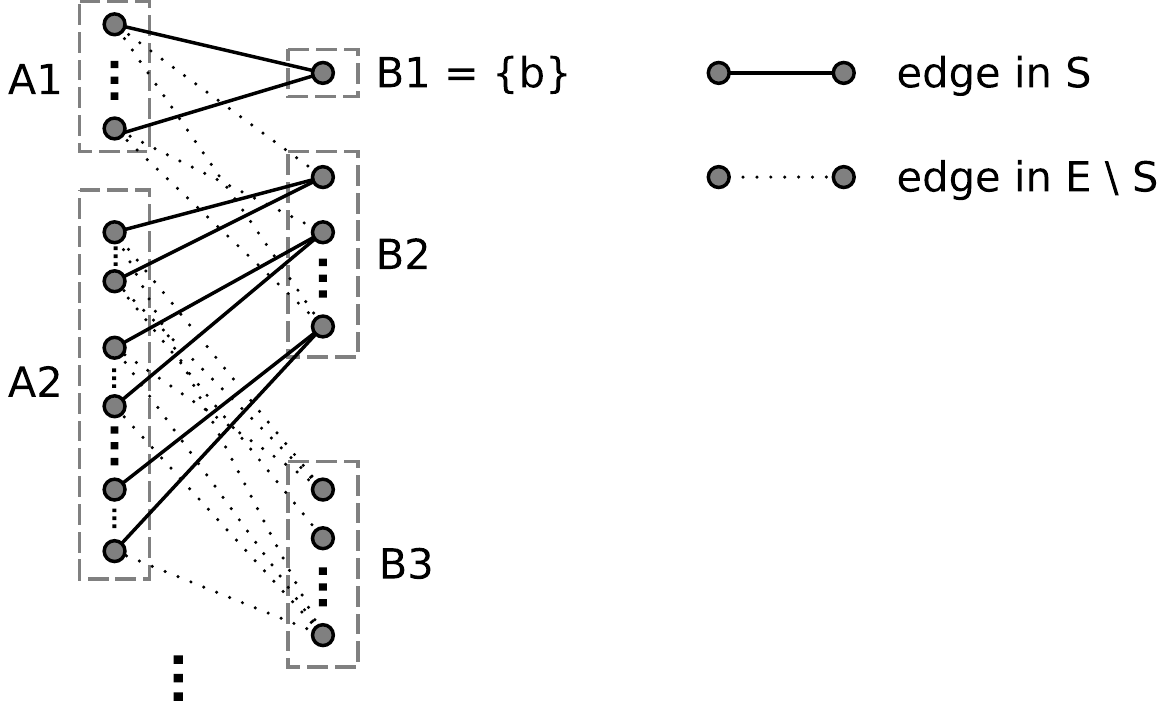}
\end{center}
\caption{Illustration of the proof of Lemma~\ref{lemma:fractional-matching}. All nodes $b' \in \bigcup_{i \ge 2} B_i$ 
have $\deg_S(b') \ge \deg_S(b) - 1$ since otherwise there is a degree-minimizing path. To keep the figure simple,
only those edges of $E \setminus S$ are drawn that connect the $A_i$ to $B_{i+1}$. Note that in general there are also 
edges outside $S$ from $A_i$ to $\bigcup_{j< i} B_j$. However, there are no edges in the graph from 
$A_i$ to $\bigcup_{j \ge i+2} B_j$. \label{figure:lemma-fractional-matching}}
\end{figure}

\begin{lemma} \label{lemma:exp-max-degree}
 Let $G = (A, B, E)$ with $|A| = n$, and let $d = \degmax \semi(A, B, E)$. 
Let $A'$ be a subset of $A$ with minimal expansion $\alpha$, that is
\begin{equation*}
 A' = \argmin_{A'' \subseteq A} \frac{|\Gamma(A'')|}{|A''|}, 
\end{equation*}
 
and let $\alpha = \frac{|\Gamma(A')|}{|A'|}$. Then:

\begin{equation*}
 d = \lceil \alpha^{-1}  \rceil .
\end{equation*}
\end{lemma}

\begin{proof} We show that $d \ge \lceil \alpha^{-1}  \rceil$ and $d \le \lceil \alpha^{-1}  \rceil$ 
separately. 

\begin{enumerate}
 \item \textbf{$d \ge \lceil \alpha^{-1}  \rceil$: } 
The set $A'$ has to be matched entirely to vertices in its neighborhood. 
Therefore, there is a node $b \in \Gamma(A')$ with degree at least
$\lceil \frac{|A'|}{|\Gamma(A')|} \rceil = \lceil \alpha^{-1} \rceil.$

 \item \textbf{$d \le \lceil \alpha^{-1}  \rceil$: } We construct a 
semi-matching explicitly with maximal degree $d$. Since an optimal
semi-matching has at most this degree, the claim follows.

Consider a decomposition of $A$ into sets $A_1, A_2, \dots$ as follows.
$A_1 \subseteq A$ is a set with minimal expansion, and for $i > 1$, 
$A_i \subseteq A \setminus (\bigcup_{j < i} A_j)$ is the set with minimal 
expansion in $G|_{(A \setminus \bigcup_{j < i} A_j) \times (B \setminus \Gamma(\bigcup_{j < i} A_j)) }$.

We construct a semi-matching $\tilde{S} = S_1 \cup S_2 \dots$ as follows. 
Firstly, match $A_1$ to $\Gamma(A_1)$ in $S_1$. By Lemma~\ref{lemma:fractional-matching}, 
the maximal degree in $S_1$ is at most 
$\lceil \frac{|A_1|}{|\Gamma(A_1)|} \rceil = \lceil \alpha^{-1}  \rceil$.

For a general $S_i$, we match $A_i$ to vertices in 
$\Gamma(A_i) \setminus \Gamma(\bigcup_{j < i} A_j)$. 
By Lemma~\ref{lemma:fractional-matching}, the maximal degree in 
$S_i$ is at most $\lceil \frac{|A_i|}{|\Gamma(A_i) \setminus \Gamma(\bigcup_{j < i} A_j)|} \rceil$.

This decomposition is illustrated in Figure~\ref{figure:lemma-exp-max-degree}.

Furthermore, it holds 

\begin{equation*}
 \frac{|A_i|}{\Gamma(A_i) \setminus \Gamma(\bigcup_{j < i} A_j)|} \le \frac{|A_{i+1}|}{\Gamma(A_{i+1}) \setminus \Gamma(\bigcup_{j < {i+1}} A_j)|},
\end{equation*}
since if this was not true, then the set $A_i \cup A_{i+1}$ would have smaller expansion in the graph 
$G|_{(A \setminus \bigcup_{j < i} A_j) \times (B \setminus \Gamma(\bigcup_{j < i} A_j)) }$
than $A_i$. This implies that $\degmax \tilde{S} = \degmax S_1$ which in turn is $\lceil \alpha^{-1}  \rceil$.
\end{enumerate}
\qed
\end{proof}

\begin{figure}[ht!]
\begin{center}
 \includegraphics[height=3.6cm]{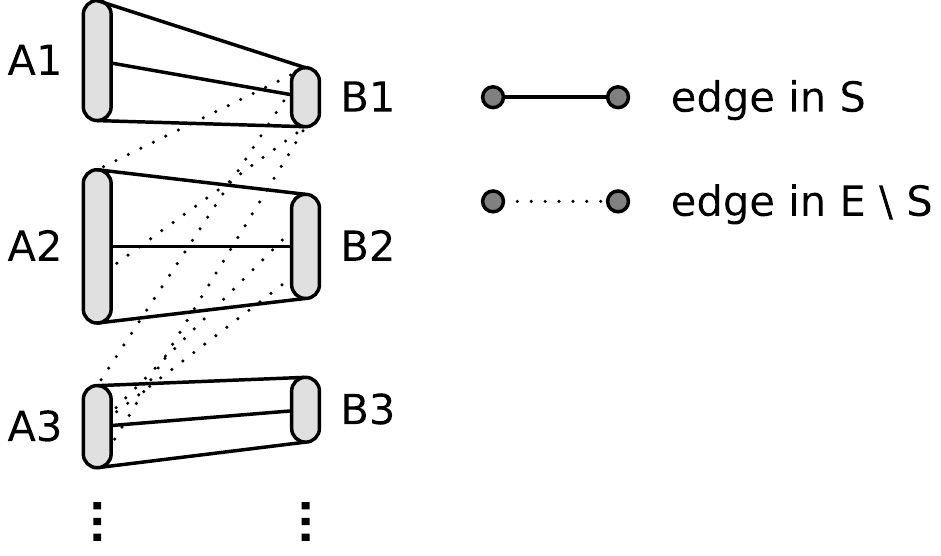}
\end{center}
\caption{Illustration of the graph decomposition used in the proof of Lemma~\ref{lemma:exp-max-degree}. 
Here, $B_i$ is the set $\Gamma(A_i) \setminus \Gamma(\bigcup_{j < i} A_j)$. The neighborhood of
$A_i$ in $G$ is a subset of $\bigcup_{j \le i} B_i$. In $S$, however, $A_i$ is matched entirely to vertices in $B_i$.
 \label{figure:lemma-exp-max-degree}}
\end{figure}

We prove now that an optimal semi-matching is a $\Order(\sqrt{n})$-skeleton.

\begin{lemma} \label{lemma:upper-bound-sparsifier-1}
 Let $G = (A, B, E)$ with $n = |A|$, and let $S = \semi(A, B, E)$. Then: 
\begin{equation*}
\forall A' \subseteq A: \degmax \semi(A', B, S) < \sqrt{n} \, (\degmax \semi(A', B, E))^{1/2} + 1. 
\end{equation*}
\end{lemma}
\begin{proof}
 Let $A' \subseteq A$ be an arbitrary subset. Let $A'' = \argmin_{A''' \subseteq A'} \frac{|\Gamma_S(A''')|}{|A'''|}$, and let
$k = |\Gamma_S(A'')|$. 
Let $d = \degmax \semi(A', B, S)$. Then by Lemma~\ref{lemma:exp-max-degree}, $d = \lceil \frac{|A''|}{k} \rceil$.
Furthermore, since $A''$ is the set of minimal expansion in $S$, for all $b \in \Gamma_S(A''): \deg_S(b) = d$,
and hence $|A''| = kd$.

Let $d^* = \degmax \semi(A'', B, E)$. Then $d^* \le \degmax \semi(A', B, E)$, since $A'' \subseteq A'$. 
% We show now that $|X| \le \frac{n-k}{d-1}$.
%For $1 \le i \le k$ let $A''_i$ be the subset of nodes of $A''$ that are matched to node $b_i$ in $S$. Let $x_i = |\Gamma_E(A''_i)|$. 
It holds that $\forall x \in \Gamma_E(A'') \setminus \Gamma_S(A''): \deg_S(x) \ge d - 1$
since otherwise there was a degree-minimizing path of length $2$ in $S$. 
Figure~\ref{figure:comm-upper-bound-1} illustrates this setting. 
The sum of the degrees of the vertices in $\Gamma_E(A'')$ is upper-bounded by the number of $A$ nodes. We obtain hence
$(|\Gamma_E(A'')| - k) (d-1) + k d \le n,$
and this implies that $|\Gamma_E(A'')| \le \frac{n-k}{d-1}$.  Clearly, $d^* \ge |A''|/|\Gamma_E(A'')|$, and using the prior upper bound
on $|\Gamma_E(A'')|$ and the equality $|A''| = kd$, we obtain $d^* \ge \frac{kd(d-1)}{n-k}$ which implies that 
$d < \sqrt{n} \sqrt{d^*} + 1$ for any $k \ge 1$. %This proves the lemma. 
\qed
\end{proof}

\begin{figure}[ht!]
 \begin{center}
\begin{minipage}{6cm}
\hspace{1.05cm} \includegraphics[height=5.5cm]{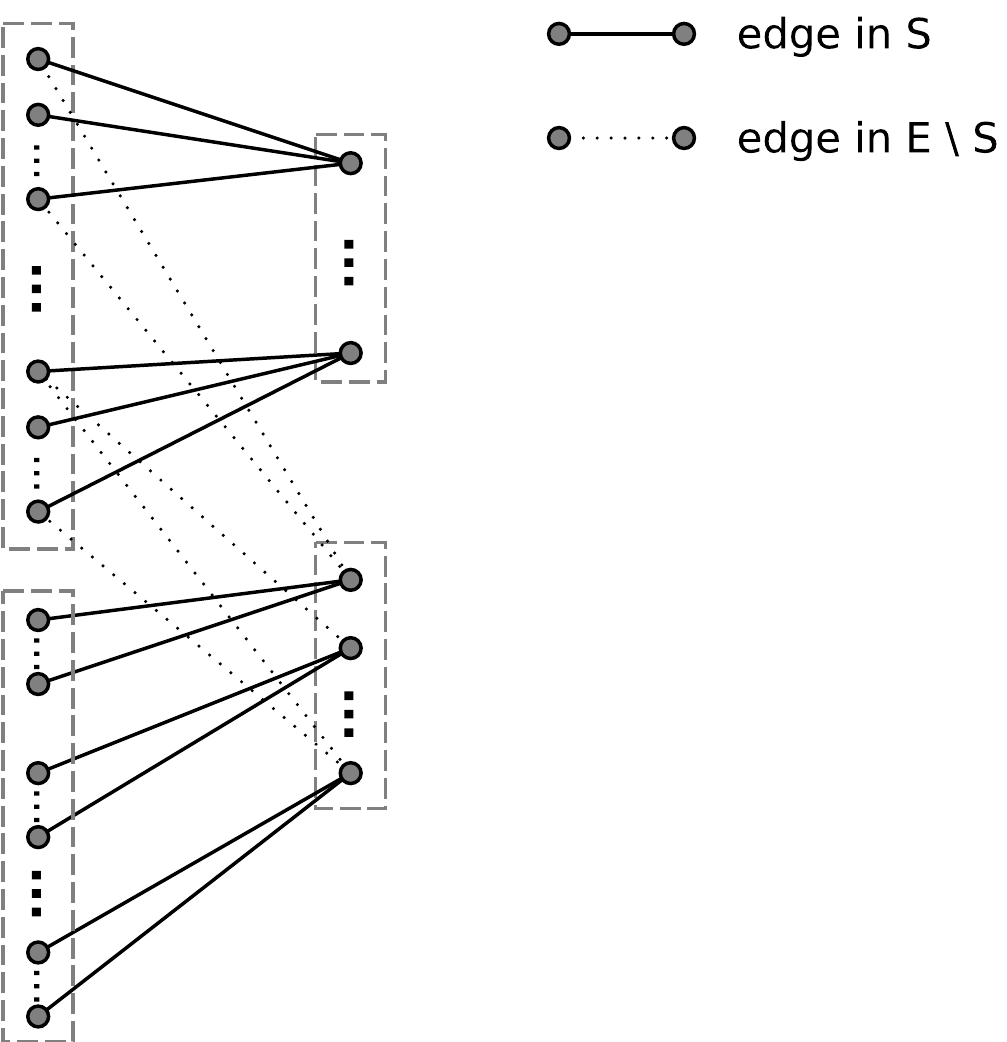}

\vspace{-4.25cm} \hspace{0.3cm} $A''$ \hspace{2.3cm} $\Gamma_S(A'')$ 

\vspace{2cm} $A \setminus A''$ \hspace{2.25cm} $\Gamma_E(A'') \setminus \Gamma_S(A'')$

\vspace{1.8cm}
\end{minipage}
\end{center}
\caption{Illustration of the proof of Lemma~\ref{lemma:upper-bound-sparsifier-1}. 
All nodes $b \in \Gamma_S(A'')$ have $\deg_S(b) = d$, and all nodes $b' \in \Gamma_E(A'') \setminus \Gamma_S(A'')$
have $\deg_S(b) \ge d-1$.
\label{figure:comm-upper-bound-1}}
\end{figure}

In order to obtain an $\Order(n^{1/3})$-skeleton, for each $a \in A$ we add one edge to the $\Order(\sqrt{n})$-skeleton.
Let $S = \semi(A, B, E)$ be the $\Order(\sqrt{n})$-skeleton, let $B' = B(S)$ be the $B$ nodes that are matched
in the skeleton, and for all $b \in B'$ let $A_b = \Gamma_S(b)$ be the set of $A$ nodes that are matched to $b$ in $S$.
Intuitively, in order to obtain a better skeleton, we have to increase the size of the neighborhood in the skeleton of all
subsets of $A$, and in particular of the subsets $A_b$ for $b \in B'$. We achieve this by adding additional optimal semi-matchings 
$S_b = \semi(A_b, B, E)$ for all subsets $A_b$ with $b \in B'$ to $S$, see Lemma~\ref{lemma:upper-bound-sparsifier-2}.
We firstly prove a technical lemma, Lemma~\ref{lemma:size-condition-semi-matching-subset}, that points out an important 
property of the interplay between the matchings $S$ and the matchings $S_b$ for $b \in B'$. Then, we state in Lemma~\ref{lemma:hoelder} 
an inequality that is an immediate consequence of H\"older's inequality. Lemma~\ref{lemma:hoelder} is then used in the proof
of Lemma~\ref{lemma:upper-bound-sparsifier-2}, which proves that our construction is an $\Order(n^{1/3})$-skeleton.

\begin{lemma} \label{lemma:size-condition-semi-matching-subset}
 Let $G = (A, B, E)$, $A' \subseteq A$, $A'' \subseteq A'$, and let $S = \semi(A', B, E)$. Furthermore, let 
$\Gamma_S(A') = \{ b_1, \dots, b_k \}$, and $\forall b_i \in \Gamma_S(A'): $ let $A'_i = \Gamma_S(b_i) \cap A'$, 
and $A''_i = \Gamma_S(b_i) \cap A''$. Then:

\begin{equation*}
 \degmax \semi(A'', B, E)^{-1} \sum_{i: b_i \in \Gamma_S(A'')} |A''_i| (|A'_i| - 1) \le |A'|.
\end{equation*}
\end{lemma}

\begin{proof}
Let $S'' = \semi(A'', B, E)$, and denote $d = \degmax S''$. Clearly,

\begin{equation}
 \sum_{b'' \in B(S'')} \deg_S(b'') \le |A'|. \label{eqn:391}
\end{equation}

Consider any $b'' \in B(S'')$. We bound $\deg_S(b'')$  from above as follows
%Then for any $b'' \in B(S'')$:
\begin{eqnarray}
\deg_{S}(b'') & \ge &  \max \{|A'_i| - 1 \, : \, \exists a \in A''_i \mbox{ with } b'' \in \Gamma_E(a) \} . \label{eqn:582}
\end{eqnarray}

Let $j$ be such that $|A'_j| - 1$ poses the maximum of the set in the right hand side of Inequality~\ref{eqn:582}. 
Note that if Inequality~\ref{eqn:582} was not true, then there would be a length two degree minimizing path in $S$ connecting $b''$ and
$b_j$. The setup up visualized in Figure~\ref{figure:comm-upper-bound-lemma}. We bound now the right hand side of Inequality~\ref{eqn:582} as follows
\begin{eqnarray}
(|A'_j| - 1) = \max \{|A'_i| - 1 \, & : & \, \exists a \in A''_i \mbox{ with } b'' \in \Gamma_E(a) \}   \nonumber \\ 
& \ge & \sum_{a \in \Gamma_{S''}(b'')} \frac{1}{\deg_{S''}(b'')} (|A'_{B(S(a))}|-1). \label{eqn:116}
\end{eqnarray}

We used here that $|A'_{B(S(a))}| \le |A'_j|$ for any $a \in \Gamma_{S''}(b'')$, and $|a \in \Gamma_{S''}(b'')| = \deg_{S''}(b'')$.
Since $d = \degmax S''$, and using Inequalities~\ref{eqn:582} and \ref{eqn:116} we obtain
\begin{equation}
\deg_{S}(b'') \ge \sum_{a \in \Gamma_{S''}(b'')} \frac{1}{d} (|A'_{B(S(a))}|-1) . \label{eqn:918}
\end{equation}

We combine Inequalities~\ref{eqn:391} and \ref{eqn:918}, and the result follows

\begin{eqnarray*}
 |A'| \ge \sum_{b'' \in B(S'')} \deg_S(b'') & \ge & \sum_{b'' \in B(S'')} \sum_{a \in \Gamma_{S''}(b'')} \frac{1}{d} (|A'_{B(S(a))}|-1)  \\
  & = & \frac{1}{d} \sum_{A''_i} |A''_i| |A'_i - 1|.
\end{eqnarray*} \qed
\end{proof}

\begin{figure}[ht!]
 \begin{center}
\includegraphics[height=5.5cm]{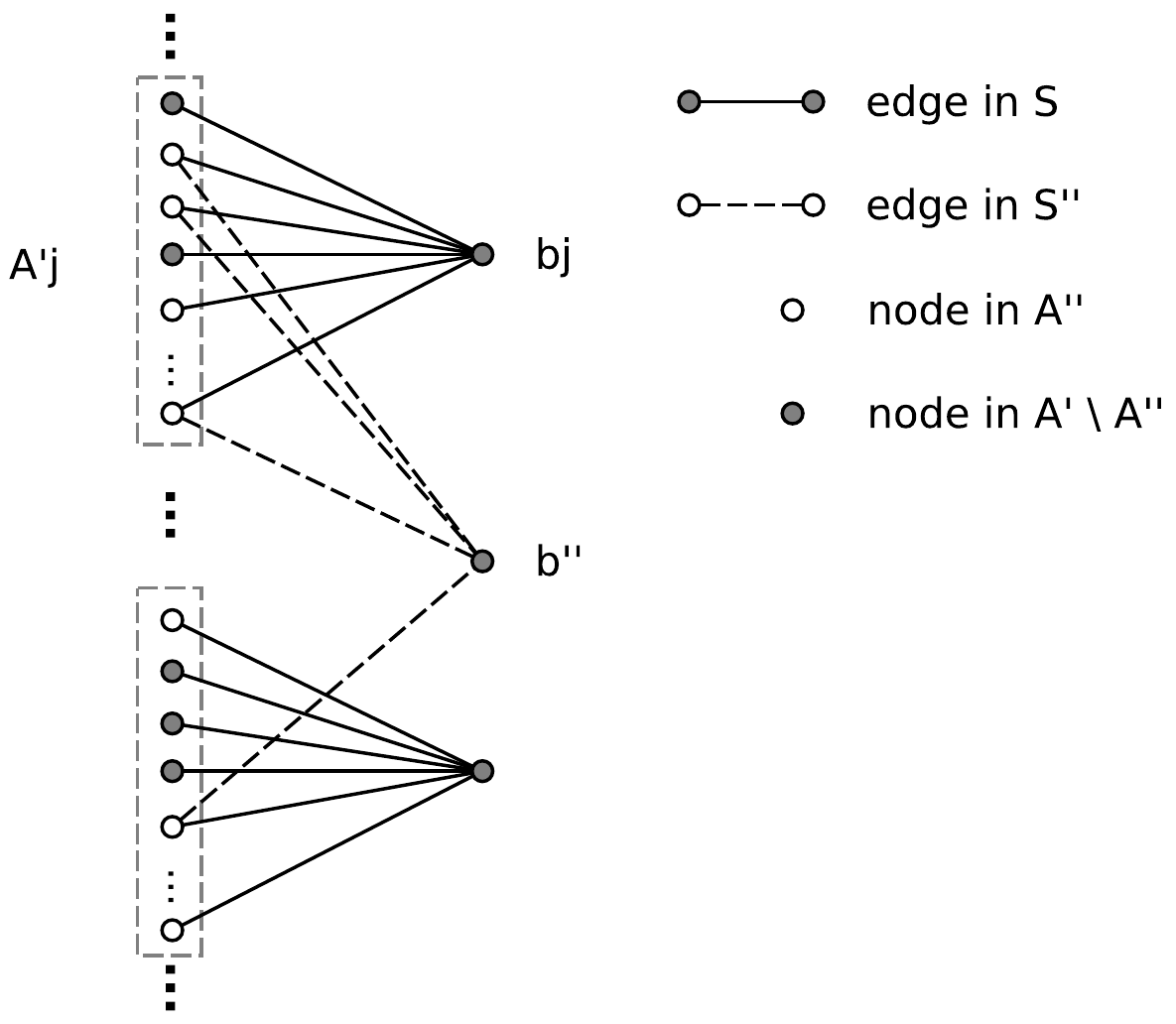}
\end{center}
\caption{Illustration of the proof of Lemma~\ref{lemma:size-condition-semi-matching-subset}. 
The degree of $b''$ in $S$ is at least $|A'_j| - 1$. Otherwise there would be a length two degree-minimizing path between $b''$ and $b_j$. 
\label{figure:comm-upper-bound-lemma}}
\end{figure}

In the proof of Lemma~\ref{lemma:upper-bound-sparsifier-2}, we also need the following 
inequality.% that is a consequence of H\"{o}lder's inequality.

\begin{lemma}  \label{lemma:hoelder}
 Let $x_1, \dots, x_k \ge 0$, and let $p > 0$ be an integer. Then:
\begin{equation*}
 \frac{(\sum_{i=1}^k x_i)^p}{k^{p-1}} \le \sum_{i=1}^k x_i^p.
\end{equation*}
\end{lemma}
\begin{proof}
 This is an immediate consequence of H\"{o}lder's inequality:

\begin{eqnarray*}
 \sum_{i=1}^k x_i \le ( \sum_{i=1}^k x_i^{p})^{1/p}  k^{\frac{p-1}{p}}.
\end{eqnarray*}
\qed
\end{proof}

\begin{lemma} \label{lemma:upper-bound-sparsifier-2}
 Let $G = (A, B, E)$ be a bipartite graph with $n = |A|$. Let $S = \semi(A, B, E)$, and for all $b \in B(S): S_b = \semi(\Gamma_S(b), B, E)$. Then:
\begin{equation*}
\forall A' \subseteq A: \degmax \semi(A', B, S \cup \bigcup_{b \in B(S)} S_b) \le \lceil 2 n^{1/3} \degmax \semi(A', B, E) \rceil. 
\end{equation*}
\end{lemma}

\begin{proof}
 Let $A' \subseteq A$. Let $\tilde{S} = S \cup \bigcup_{b \in B(S)} S_b$. Let 
$A'' = \argmin_{A''' \subseteq A'} \frac{|\Gamma_{\tilde{S}}(A''')|}{|A'''|}$ and let $k = |\Gamma_{\tilde{S}}(A'')|$.
From Lemma~\ref{lemma:exp-max-degree} it follows that $\degmax \semi(A', B, \tilde{S}) = \lceil \frac{|A''|}{k} \rceil$.
Furthermore, let $d = \degmax \semi(A'', B, E)$.
For a node $b \in \Gamma_{\tilde{S}}(A'')$, let $A''_b = \{a \in A \, : \, \tilde{S}(a) = b \}$. For two nodes 
$b_i, b_j \in \Gamma_{\tilde{S}}(A'')$, let $A''_{b_i, b_j} = \{ a \in A'' \, : \, S(a) = b_i, S_{b_i}(a) = b_j \}$.

We consider the cases $k \ge n^{1/3}$ and $k < n^{1/3}$ separately.

\begin{enumerate}
 \item $k \ge n^{1/3}$. Consider the semi-matching $S$. From Lemma~\ref{lemma:size-condition-semi-matching-subset} we obtain the 
condition
\begin{equation*} 
 1/d \sum_{i=1}^k |A''_i| (A_i - 1) \le n,
\end{equation*}
and since $A''_i \le A_i$ we obtain from the prior Inequality that
\begin{equation*}
 1/d \sum_{i=1}^k (|A''_i| - 1)^2 < n.
\end{equation*}

Using $\sum_{i=1}^k |A''_i| = |A''|$ and Lemma~\ref{lemma:hoelder}, we obtain 
\begin{eqnarray}
\nonumber \frac{1}{d}  \frac{1}{k}  (|A''| - k)^2 & < & n, \quad \Rightarrow \\
 |A''| & < & \sqrt{n d k} + k . \label{eqn:293}
\end{eqnarray}
Then, since $\degmax \semi(A'', B, \tilde{S}) = \lceil \frac{|A''|}{k} \rceil$, we obtain from Inequality~\ref{eqn:293}
$\degmax \semi(A'', B, \tilde{S}) \le \lceil \frac{\sqrt{n d}}{\sqrt{k}} \rceil + 1$. 
Since $k \ge n^{1/3}$,
we conclude that
\begin{equation*}
 \degmax \semi(A'', B, \tilde{S}) \le n^{1/3} \sqrt{d} + 2.
\end{equation*}

 \item $k < n^{1/3}$. We consider here the two subcases $|A''| < 2dk^2$ and $|A''| \ge 2dk^2$.
\begin{enumerate}
\item $|A''| < 2dk^2$. Then since $\degmax \semi(A'', B, \tilde{S}) = \lceil \frac{|A''|}{k} \rceil$, 
we conclude that
\begin{equation*}
\degmax \semi(A'', B, \tilde{S}) \le \lceil 2dk \rceil < \lceil 2d n^{1/3} \rceil .
\end{equation*}

\item $|A''| \ge 2dk^2$.
Let $b \in B(S)$ and consider the semi-matching $S_b$ matching $A''_b$ to $B$. 
From Lemma~\ref{lemma:size-condition-semi-matching-subset} and the fact that $A''_{b, b_i} \subseteq A'_{b, b_i}$ 
we obtain

\begin{eqnarray*}
 \frac{1}{d} \sum_{i=1}^k |A''_{b, b_i}| (|A''_{b, b_i}| - 1) & \le & |A_b|, \\
 \left( \frac{1}{d} \sum_{i=1}^k |A''_{b, b_i}|^2 \right) - \frac{1}{d} |A''_b| & \le & |A_b| .
\end{eqnarray*}
By Lemma~\ref{lemma:hoelder}, we obtain
\begin{equation}
\frac{1}{dk} |A''_b|^2 - \frac{1}{d} |A''_b| \le |A_b| .  \label{eqn:984}
\end{equation}
Consider now the semi-matching $S$. From Lemma~\ref{lemma:size-condition-semi-matching-subset} we obtain the 
condition
\begin{equation} 
 \frac{1}{d} \sum_{i=1}^k |A''_i| (|A_i| - 1) \le n . \label{eqn:200}
\end{equation}

Using Inequality~\ref{eqn:984} in Inequality~\ref{eqn:200} and simplifying, we obtain
\begin{eqnarray}
 \nonumber \frac{1}{d} \sum_{i=1}^k |A''_i| \left( (\frac{1}{dk }|A''_i|^2 - \frac{1}{d} |A''_i|) - 1 \right) & \le & n, \\
 \nonumber \frac{1}{d^2 k} \sum_{i=1}^k |A''_i|^3 - \sum_{i=1}^k \frac{1}{d^2} |A''_i|^2 - \sum_{i=1}^k \frac{1}{d} |A''_i| & \le & n, \\
 \frac{1}{d^2 k^3} |A''|^3 - \underbrace{\frac{1}{d^2 k} |A''|^2}_{I} - \underbrace{\frac{1}{d} |A''|}_{II} & \le & n. \label{eqn:451}
\end{eqnarray}

Since $|A''| \ge 2dk^2$, we can upper bound the terms $I$ and $II$ from Inequality~\ref{eqn:451} as follows
\begin{eqnarray}
  \frac{1}{2 d^3 k^3} |A''|^3 \ge I , \mbox{ and } \label{eqn:512} \\
  \frac{1}{4 d^3 k^4} |A''|^3 \ge II .\label{eqn:513}
\end{eqnarray}
Using bounds \ref{eqn:512} and \ref{eqn:513} in Inequality~\ref{eqn:451} and simplifying, we obtain
\begin{eqnarray}
 \nonumber \frac{1}{4d^2 k^3} |A''|^3 < n, \Rightarrow \\
 |A''| < 2^{2/3} n^{1/3} d^{2/3} k.  \label{eqn:481}
\end{eqnarray}

Since $\degmax \semi(A'', B, \tilde{S}) = \lceil \frac{|A''|}{k} \rceil$, and using Inequality~\ref{eqn:481}, we conclude that
\begin{equation*}
\degmax \semi(A'', B, \tilde{S}) \le \lceil 2^{2/3} n^{1/3} d^{2/3} \rceil. 
\end{equation*}
\end{enumerate}
\end{enumerate}

Combining the bounds from cases 1, 2a and 2b, the result follows. \qed
\end{proof}

We mention that there are graphs for which adding further semi-matchings $S_{b_1 b_2} = \semi(A_{b_1 b_2}, B, E)$ 
to our $\Order(n^{1/3})$-skeleton, where 
$A_{b_1 b_2}$ is the set of $A$ vertices whose neighborhood in our $\Order(n^{1/3})$-skeleton is the set 
$\{ b_1, b_2 \}$, does not help to improve the quality of the skeleton. 
Before stating our main theorem, Theorem~\ref{thm:communication-ub}, we show in Lemma~\ref{lemma:combining-semi-matchings} 
that if Alice sends a $c$-matching skeleton, then Bob can compute a $c+1$ approximation. Then, we state our main theorem.

%Proof in Appendix.
%\begin{lemma} \label{lemma:combining-semi-matchings}
% Let $G = (A, B, E)$ be a bipartite graph and let $E_1, E_2$ be a partition of the edge set $E$. Furthermore, 
%let $E'_1 \subseteq E_1$ such that for any $A' \subseteq A(E_1)$:
%\begin{equation*}
%\degmax \semi(A(E_1), B, E'_1) \le c \degmax \semi(A(E_1), B, E'_1). 
%\end{equation*}
%\begin{equation*}
%\mbox{Then:} \quad \quad \quad \degmax \semi(A, B, E'_1 \cup E_2) \le (c+1) \degmax \semi(A, B, E).
%\end{equation*}
%\end{lemma}

\begin{lemma} \label{lemma:combining-semi-matchings}
 Let $G = (A, B, E)$ be a bipartite graph and let $E_1, E_2$ be a partition of the edge set $E$. Furthermore, 
let $E'_1 \subseteq E_1$ such that for any $A' \subseteq A(E_1)$:
\begin{equation*}
\degmax \semi(A(E_1), B, E'_1) \le c \degmax \semi(A(E_1), B, E'_1). 
\end{equation*}
Then:
\begin{equation*}
\degmax \semi(A, B, E'_1 \cup E_2) \le (c+1) \degmax \semi(A, B, E).
\end{equation*}
\end{lemma}

\begin{proof}
 We construct a semi-matching $S$ between $A$ and $B$ with edges from $E'_1 \cup E_2$ explicitly and
we show that $\degmax S \le (c+1) \degmax \semi(A, B, E)$. Since $\degmax \semi(A, B, E'_1 \cup E_2) \le \degmax S$,
the result then follows.

Let $S_2 = \semi(A, B, E) \cap E_2$, and let $S_1 = \semi(A \setminus A(S_2), B, E_1)$. Then $S = S_1 \cup S_2$.
Clearly, $\degmax S_2 \le \degmax \semi(A, B, E)$. Furthermore, by the premise of the lemma we obtain 
$\degmax S_1 \le c \degmax \semi(A, B, E)$. Since 
$\degmax S \le \degmax S_1 + \degmax S_2$ and $\degmax S_1 + \degmax S_2 \le (c+1) \degmax(A, B, E)$ the
result follows. \qed
\end{proof}

\begin{theorem} \label{thm:communication-ub}
 Let $G = (A, B, E)$ with $n = |A|$ and $m = |B|$. Then there are one-way two party deterministic communication protocols 
for the semi-matching problem, one with
\begin{enumerate}
 \item message size $c n \log m$ and approximation factor $n^{1/2} + 2$, and another one with 
 \item message size $2 c n \log m$ and approximation factor $2n^{1/3} + 2$.
\end{enumerate}
\end{theorem}

\begin{proof}
 Alice computes the skeletons as in Lemma~\ref{lemma:upper-bound-sparsifier-1} or in Lemma~\ref{lemma:upper-bound-sparsifier-2}
and sends them to Bob. Bob computes an optimal semi-matching considering his edges and the edges received from Alice. 
By Lemma~\ref{lemma:combining-semi-matchings} the results follow. \qed
\end{proof}

\subsection{Lower Bounds for Semi-matching-skeletons} 
We present now a lower bound that shows that the skeletons of the previous subsection are essentially optimal.
For an integer $c$, we consider the complete bipartite graph $K_{n, m}$ where $m$ is a carefully chosen value depending on $c$ and $n$.
We show in Lemma~\ref{lemma:graph-existence-1} that for any subset of edges $E'$ of $K_{n, m}$ such that for all 
$a \in A: \deg_{E'}(a) \le c$, there is a subset $A' \subseteq A$ with $|A'| \le m$ such that an optimal semi-matching
that matches $A'$ using edges in $E'$ has a maximal degree of $\Omega(n^{\frac{1}{c+1}})$. Note that since
$|A'| \le m$, there is a matching in $K_{n,m}$ that matches all $A'$ vertices. This implies that such an $E'$
is only an $\Omega(n^{\frac{1}{c+1}})$-skeleton.

\begin{lemma} \label{lemma:graph-existence-1}
 Let $G = (A, B, E)$ be the complete bipartite graph with $|A| = n$ and $|B| = (c!)^{\frac{1}{c+1}} n^{\frac{1}{c+1}}$ for an integer 
$c$. Let $E' \subseteq E$ be an arbitrary subset such that $\forall a \in A: \deg_{E'}(a) \le c$. 
Then there exists an $A' \subseteq A$ with $|A'| \le |B|$ and 
\begin{eqnarray} 
\degmax \semi(A', B, E') \ge  \frac{(c!)^{\frac{1}{c+1}}}{c} n^{\frac{1}{c+1}} > e^{-1.3} n^{\frac{1}{c+1}}. \label{eqn:102}
\end{eqnarray}
\end{lemma}

\begin{proof}
 Let $E' \subseteq E$ be as in the statement of the lemma. Let $E''$ be an arbitrary superset of $E'$ such that 
 $\forall a \in A: \deg_{E''}(a) = c$. Since $\degmax \semi(A', B, E'') \le \degmax \semi(A', B, E')$ 
 it is enough to show the lemma for $E''$. Denote by $A_{ \{i_1, \dots, i_c \}}$ the subset of $A$ such that 
$\forall a \in A_{ \{i_1, \dots, i_c \}}: \Gamma_{E''}(a) = \{ b_{i_1}, \dots, b_{i_c} \}$. Then

\begin{equation} 
 |A| = \sum_{\substack{A_i: i = \{ i_1, \dots, i_c \} \text{ and } \\ \{b_{i_1}, \dots, b_{i_c} \} \text{ is a $c$-subset of $B$ } } } |A_i|, \label{eqn:2392}
\end{equation}

since $\forall a \in A: \deg_{E''}(a) =c$. Suppose for the sake of a contradiction that Inequality~\ref{eqn:102} is not 
true. Then for all $A_i$ on the right side of Inequality~\ref{eqn:2392} we have $|A_i| < (c!)^{\frac{1}{c+1}} n^{\frac{1}{c+1}}$. 
There are at most $\binom{|B|}{c}$ such sets. This implies that:
\begin{eqnarray*}
 |A| & \le & \binom{|B|}{c} \cdot (c!)^{\frac{1}{c+1}} n^{\frac{1}{c+1}} < \frac{|B|^c}{c!}  (c!)^{\frac{1}{c+1}} n^{\frac{1}{c+1}}  
 <  \frac{(c!)^{\frac{c}{c+1}} n^{\frac{c}{c+1}}}{c!} (c!)^{\frac{1}{c+1}} n^{\frac{1}{c+1}} = n.
\end{eqnarray*}
 This is a contradiction to the fact that $|A| \ge n$ and proves the first inequality in Inequality~\ref{eqn:102}. 
To proof the second, we apply Stirling's formula, and we obtain

\begin{eqnarray*}
 \frac{(c!)^{\frac{1}{c+1}}}{c} > \frac{(\sqrt{2 \pi} c^{c+1/2} e^{-c})^{\frac{1}{c+1}} }{c} = e^{ \frac{1/2 \ln (2 \pi) - 1/2 \ln(c) - c }{c+1} } . \label{eqn:439}
\end{eqnarray*}

It can be shown that for any $c > 0$, $\frac{1/2 \ln (2 \pi) - 1/2 \ln(c) - c }{c+1} > - 1.3$ which proves the result.
\qed
% short version
%The second inequality
%follows from the fact that for all $c > 0$, $\frac{(c!)^{\frac{1}{c+1}}}{c} > e^{-1.3}$. \qed
\end{proof}

We extend Lemma~\ref{lemma:graph-existence-1} now to edge sets of bounded cardinality without restriction on 
the maximal degree of an $A$ node, and we state then our lower-bound result in Theorem~\ref{thm:lb-skeleton}.

\begin{lemma} \label{lemma:graph-existence-2}
 Let $c > 0$ be an integer, let $\epsilon > 0$ be a constant, and let $c' = (1+\epsilon)c$. Let $G = (A, B, E)$ be the complete bipartite graph with $|A| = n$ 
and $|B| = (c'!)^{\frac{1}{c'+1}} (\frac{\epsilon}{1+\epsilon}\cdot n)^{\frac{1}{c'+1}}$. 
 Let $E' \subseteq E$ be an arbitrary subset of size at most $c \cdot n$. Then there exists an $A' \subseteq A$ with $|A'| \le |B|$ and 

\begin{eqnarray} 
\degmax \semi(A', B, E') > e^{-1.3} (\frac{\epsilon}{1+\epsilon}n)^{\frac{1}{c'+1}}. \label{eqn:105}
\end{eqnarray}
\end{lemma}

\begin{proof}
 Split $A$ into $A_{>}$ and $A_{\le}$ such that for all $a \in A_{>}: \deg_{S'}(a) > c'$, and for all 
 $a \in A_{\le}: \deg_{S'}(a) \le c'$.
 Then $|A_{>}| c' + |A_{\le}| \le c n$ which implies that $|A_{\le}| \ge \frac{\epsilon}{1+\epsilon} n$. 
Let $G' = G|_{A_{\le} \times B}$. Then by Lemma~\ref{lemma:graph-existence-1} applied on $G'$
there is a subset $A' \subseteq A_{\le}$ with $|A'| \le |B|$ such that 
\begin{equation*}
 \degmax \semi(A', B, E'|_{A_{\le} \times B}) > e^{-1.3} |A_{\le}|^{\frac{1}{c'+1}},
\end{equation*}
and since $\degmax \semi(A', B, E'|_{A_{\le} \times B}) = \degmax \semi(A', B, E')$, the result follows. \qed
\end{proof}

\begin{theorem} \label{thm:lb-skeleton}
 Let $c > 0$ be an integer. Then for all $\epsilon > 0$, an $\Order(n^\frac{1}{(1+\epsilon)c + 1})$-semi-matching skeleton requires
at least $cn$ edges.
\end{theorem}

%%%%% Short version 
%We consider a one-way two-party protocol which is given as input a bipartite graph $G=(A, B,E)$.
%We assume that the sets $A$ and $B$ are known in advance to both Alice and Bob (this only strengthens 
%our lower bound) and that  $E_A \subseteq E$ is given to Alice and  $E_B \subseteq E$ is given to Bob. 
%We do not impose any particular format on the output of the protocol. We only require that Bob
%gives an output that maps to a valid semi-matching for $S$ for $G$, i.e., $S \subseteq E$.

%%%%%% Long version 
%We consider a one-way two-party protocol which is given as input a bipartite graph $G=(A, B,E)$.
%We assume that the sets $A$ and $B$ are known in advance to both Alice and Bob (this only strengthens 
%our lower bound) and that the set of edges $E$ is partitioned into two sets $E_A$ and $E_B$,
%where $E_A$ is given to Alice and $E_B$ is given to Bob. We do not impose any particular structure 
%or format on the output of the protocol. We only require that Bob, the player receiving the single message
%in the protocol, gives an output that maps to a valid semi-matching for $G$, i.e.,
%a set $S$ of $|A|$ edges, one incident to each of the nodes of $A$, and all of them 
%edges in $E$ (i.e., $S \subseteq E$).
%

%%%%% short version 
\subsection{One-way, two party communication lower bound} 
To prove a lower bound on the deterministic communication complexity we define 
a family of bipartite graphs. For given integers $n$ and $m$, 
let $\mathcal{G}_1 = \{ G_1(x) |  x \in \{0,1\}^{n\times m}\}$ be defined as 
follows. 
Let $B_0 = \{b^0_1, \dots, b^0_m \}$, $B_1 = \{b^1_1, \dots, b^1_m \}$ and
 $A =\{a_1, \dots, a_n \}$.
Given $x \in \{0,1\}^{n \times m}$, let $E_x = \{ (a_i, b^{x_{i,j}}_j) \, | \, 1\leq i \leq n, 1 \leq j \leq m \}$
(i.e, the entries of the matrix $x$ determine if there is an edge $(a_i, b^0_j)$  or
an edge  $(a_i, b^1_j)$ for all $i,j$). Then, we define $G_1(x)=(A , B_0 \cup B_1, E_x)$. 
%
%
%
%%%%%%%%long version 
%To prove the lower bound we define a family of graphs. For given integers $n$ and $m$, 
%let the family of graphs $\mathcal{G}_1 = \{ G_1(x) |  x \in \{0,1\}^{n\times m}\}$ be defined as 
%follows. Let $B_0 = \{b^0_1, \dots, b^0_m \}$ and $B_1 = \{b^1_1, \dots, b^1_m \}$ be two 
%sets of $m$ nodes each, and let $A =\{a_1, \dots, a_n \}$ be a set of $n$ nodes. 
%Given $x \in \{0,1\}^{n \times m}$, let $E_x = \{ (a_i, b^{x_{i,j}}_j) \, | \, 1\leq i \leq n, 1 \leq j \leq m \}$
%(i.e, the entries of the matrix $x$ determine if there is an edge $(a_i, b^0_j)$  or
%an edge  $(a_i, b^1_j)$ for all $i,j$). Then, we define $G_1(x)=(A , B_0 \cup B_1, E_x)$. 
From  Lemma~\ref{lemma:graph-existence-2} we immediately obtain the following lemma. 

\begin{lemma} \label{lemma:graph-existence-3}
 Let $c > 0$ be an integer, let $\epsilon > 0$ be a constant, and let $c' = (1+\epsilon)c$. Let $n$ be a sufficiently large integer,
 and let $m = (c'!)^{\frac{1}{c'+1}} (\frac{\epsilon}{1+\epsilon}\cdot n)^{\frac{1}{c'+1}}$.
 Let $G=(A, B_0\cup B_1,E)$ be a graph $G \in \mathcal{G}_1$, and 
 let $E' \subseteq E$ be such that $|E'| \leq cn$. Then there exists a set of nodes 
 $A' \subseteq A$ with $|A'| \le m$ and 
 $\degmax \semi(A', B_0 \cup B_1, E') > 1 / 2e^{-1.3} (\frac{\epsilon}{1+\epsilon}n)^{\frac{1}{c'+1}}$.

%\begin{eqnarray} 
%\degmax \semi(A', B_0 \cup B_1, E') > 1 / 2e^{-1.3} (\frac{\epsilon}{1+\epsilon}n)^{\frac{1}{c'+1}}. \label{eqn:106}
%\end{eqnarray}
\end{lemma}

%%%short version 
We further define a second family of bipartite graphs  $\mathcal{G}_2$  on the sets of nodes
$A$ and $C$, $|A|=|C|=n$.
 For a set $A' \subseteq A$  we define the graph $G_2(A')$
to be an arbitrary matching from all the nodes of  $A'$ to nodes of  $C$. 
The family of graphs  $\mathcal{G}_2$  is defined as $\mathcal{G}_2 = \{ G_2(A') | A' \subseteq A \}$.

%%%%%%long version 
%We further define a second family of graphs  $\mathcal{G}_2$ defined on a set $A$ of $n$ nodes
%and a set $C$ of $n$ nodes. For a set $A' \subseteq A$  we define the graph $G_2(A')$
%to be an arbitrary matching, matching all nodes of  $A'$ to nodes of  $C$. 
%The family of graphs  $\mathcal{G}_2$  is defined as $\mathcal{G}_2 = \{ G_2(A') | A' \subseteq A \}$. 
% 

Our lower bound will be proved using a family of graphs $\mathcal{G}$. Slightly abusing 
notation, the family of graphs  $\mathcal{G}$ is defined as 
$\mathcal{G} = \mathcal{G}_1 \times \mathcal{G}_2$. That is, the graphs in 
$\mathcal{G}$ are all graphs $G=(A,B_0 \cup B_1 \cup C, E_1 \cup E_2)$ built from  a graph
$G_1=(A , B_0  \cup B_1, E_1)  \in\mathcal{G}_1$ and a graph  $G_2=(A, C, E_1) \in \mathcal{G}_2$ 
where the set of nodes $A$  is the same for $G_1$ and $G_2$.
We now prove our lower bound.

\begin{theorem} \label{thm:communicatio-lb}
 Let $c > 0$ be an integer and  let $\epsilon>0$ be an arbitrarily small constant.
 Let ${\cal P}$ be a $\beta$-approximation one-way two-party protocol for semi matching that 
 has communication complexity at most $\alpha$. 
 If  $ \beta  \leq  \gamma = 1/2 \frac{1}{e^{1.3}} (\frac{\epsilon}{\epsilon + 1} n )^{\frac{1}{(1+\epsilon)c + 1}}$,
  then $\alpha > cn$, where $n$ is the number of nodes to be matched. 
 \end{theorem}
 \begin{proof}
Take $n$ sufficiently large. Let $c'=(1+\epsilon)c$ and let 
$m= (c'!)^{\frac{1}{c'+1}} (\frac{\epsilon}{1+\epsilon}\cdot n)^{\frac{1}{c'+1}}$.
  We consider as  possible inputs the graphs in  $\mathcal{G}$ 
  (for $n$ and $m$). 
 Given an input graph, Alice will get as input all edges between $A$ and $B_0 \cup B_1$
 (i.e., a graph in $\mathcal{G}_1$) and Bob will get all edges between $A$ and $C$ (i.e., 
 a graph in  $\mathcal{G}_2$)

 %%%%%short version
 Assume towards a contradiction that the communication complexity of ${\cal P}$ is at most
 $cn$.
Then  there is a set of graphs ${\cal G}^* \subseteq {\cal G}_1$,
$|{\cal G}^*| \geq 2^{nm-cn}$, such that on all graphs in ${\cal G}^*$  Alice sends the
 same message to Bob.  Consider  the set $X^* \subseteq \{0,1\}^{n\times m}$ such
 that  ${\cal G}^* = \{ G_1(x) \, | \, x \in X^*\}$,
 %% i.e.,  the binary matrices of size  that define  the graphs in ${\cal G}^*$. 
 Since there is a one-to-one correspondence between ${\cal G}^*$ and $X^*$,
 $|X^*| \geq 2^{nm-cn}$, and there are at most $cn$ entries which are constant 
 over all matrices in $X^*$, otherwise $|X^*| < 2^{nm-cn}$.
%%%%% long version 
% Assume towards a contradition that the communication complexity of ${\cal P}$ is at most
% $cn$.
% Since there are $2^{nm}$ graphs in $\mathcal{G}_1$, 
% it follows that there is a set of graphs ${\cal G}^* \subseteq {\cal G}_1$,
% with $|{\cal G}^*| \geq 2^{nm-cn}$, such that on any of the graphs in ${\cal G}^*$ given as input 
% to Alice, Alice sends the
% same message to Bob.  Consider now the set $X^* \subseteq \{0,1\}^{n\times m}$ such
% that  ${\cal G}^* = \{ G_1(x) \, | \, x \in X^*\}$, i.e.,  the binary matrices of size  $n\cdot m$ that define
% the graphs in ${\cal G}^*$. There is a one to one correspondace between ${\cal G}^*$ and $X^*$,
% hence $|X^*| \geq 2^{nm-cn}$. It follows that there are at most $cn$ entries which are constant 
% over all matrices in $X^*$, otherwise $|X^*| < 2^{nm-cn}$.
 %%%%short version
 This means that there are at most  $cn$ edges that exist in all graphs in ${\cal G}^*$. Let $E'$ be the
 set of all these edges. 
 
% %%%%%long version
% Now define $Y^*=\{ (i,j) \, | \, 1\leq i \leq n, 1 \leq j \leq m,$ entry $(i,j)$ is constant over all 
%matrices in $X^* \}$, and let the set of edges $E'$ be defined as  
% $E' =  \{ (a_i, b^{y_{i,j}}_j) | (i,j) \in Y^*, y_{i,j} = x_{i,j}  ~ \mbox{for some~} x \in X^*\}$ (i.e., the 
% set $E'$ is the set of edges that exist in all graphs $G \in {\cal G}^*$.)

 Consider now the graph $G=(A,B_0 \cup B_1, E')$.  Since $|E'| \leq cn$, 
 by Lemma~\ref{lemma:graph-existence-3} there exists a set  $A' \subseteq A$ with $|A'| \le m$ and 
 $\degmax \semi(A', B_0 \cup B_1, E') > \gamma$. We now define 
 $G_2^* \in {\cal G}_2$ to be $G_2^*=G_2(A \setminus A')$.
  
Now observe that on any of $G \in {\cal G}^* \times \{ G_2^* \} \subseteq {\cal G}$, 
%given as input to ${\cal P}$, 
${\cal P}$ gives the same output semi-matching $S$. 
$S$ can include, 
as edges matching 
the nodes in $A'$,  
only edges from $E'$, since for any other edge
there exists an input in $ {\cal G}^* \times  \{G_2^*\}$ in which that edge does not exist and 
$ {\cal P}$ would  not be correct on that input. It follows (by  Lemma~\ref{lemma:graph-existence-3}) 
that the maximum degree of $S$ is greater than $\gamma$. 
On the other hand, since $|A'| \leq m$, there is a perfect matching in any graph in ${\cal G}^* \times \{G_2^*\}$.
The approximation ratio of ${\cal P}$ is therefore greater than $\gamma$. A contradiction. \qed
\end{proof}

\section{The Structure of Semi-Matchings} \label{section:structure}

We now present our results concerning the structure of semi-matchings. 
%In Subsection~\ref{subsection:decomp},
%we state two lemmas that show that an optimal semi-matching can be decomposed in to maximum matchings, and
%a semi-matching that does not admit a length two degree-minimizing path can be decomposed into maximal matchings.
%We show that a semi-matching that does not admit a length two degree-minimizing path approximates the maximal
%degree of an optimal semi-matching by a factor of $\lceil \log(n + 1) \rceil$.
%In Subsection~\ref{subsection:online-algo}, as an application we discuss the deterministic online 
%semi-matching algorithm of \cite{anr95}, and we show that our results prove the competitive ratio
%of $\lceil \log(n + 1) \rceil$ of the algorithm.
%\subsection{Decomposition Lemmas} \label{subsection:decomp} 
Firstly, we show 
in Lemma~\ref{lemma:no-length-2-deg-min-path} that
a semi-matching that does not admit length $2$ degree-minimizing paths can be decomposed into
maximal matchings. In Lemma~\ref{lemma:no-deg-min-path},
we show that if a semi-matching does not admit \textit{any} degree-minimizing paths, then there is
a similar decomposition into maximum matchings.

Lemma~\ref{lemma:no-length-2-deg-min-path} is then used to prove that semi-matchings that do not
admit length $2$ degree-minimizing paths approximate optimal semi-matchings within a factor
$\lceil \log(n + 1) \rceil$. To this end, we firstly show in Lemma~\ref{lemma:matching-half-of-the-nodes} 
that the first $d^*$ maximal matchings of the decomposition of such a
semi-matching match at least $1/2$ of the $A$ vertices, where $d^*$ is the maximal degree
of an optimal semi-matching. 
In Theorem~\ref{theorem:log-n-approximation}, we then apply this result $\lceil \log(n + 1) \rceil$ times, 
showing that the maximal degree of a semi-matching that does not admit length $2$ degree-minimizing 
paths is at most $\lceil \log(n + 1) \rceil$ times the maximal degree of an optimal semi-matching.

\begin{lemma} \label{lemma:no-length-2-deg-min-path}
 Let $S = \semi_2(A, B, E)$ be a semi-matching in $G$ that does not admit a length $2$ degree-minimizing path, 
and let $d = \degmax S$. Then $S$ can be partitioned into $d$ matchings $M_1, \dots, M_d$ such that 
\begin{eqnarray*}
\forall i: M_i  \mbox{ is a maximal matching in } G|_{A_i \times B_i}, 
\end{eqnarray*}

where $A_1 = A$, $B_1 = B$, and for $i > 1: A_i = A \setminus \bigcup_{1 \le j < i} A(M_j)$ and $B_i = B(M_{j-1})$.
\end{lemma}

\begin{proof}
 The matchings $M_1, \dots, M_d$ can be obtained as follows. For each $b \in B(S)$, label its incident
edges in $S$ by $1, 2, \dots, \deg_S(b)$ arbitrarily. Matching $M_i$ is then the subset of edges of $S$ that are labeled
by $i$. 

We prove the statement by contradiction. Let $i$ be the smallest index such that $M_i$ is not maximal in 
$G|_{A_i \times B_i}$. Then there exists an edge $e = ab \in E$ with $a \in A_i$ and $b \in B_i$ 
such that $M_i \cup \{e \}$ is a matching in $G|_{A_i \times B_i}$. Note that $\deg_S(b) < i$ since $b$ is not matched in $M_i$. 
Consider now the edge $e' \in S$ matching the node $a$ to $b'$ in $S$. Since $a \in A_i$ and $a$ is not matched in $M_i$, 
$e'$ is in a matching $M_j$ with $j > i$ and hence $\deg_S(b') \ge j > i$. Then $P = (b', a, b)$
is a length $2$ degree-minimizing path since $\deg_S(b') > i$ and $\deg_S(b) < i$ contradicting our assumption. \qed
\end{proof}

\begin{lemma} \label{lemma:no-deg-min-path}
 Let $S^* = \semi(A, B, E)$ be a semi-matching in $G$ that does not admit degree-minimizing paths of any length, and let $d^* = \degmax S^*$. 
Then $S^*$ can be partitioned into $d^*$ matchings 
$M_1, \dots, M_{d^*}$ such that
\begin{eqnarray*}
\forall i: M_i  \mbox{ is a maximum matching in } G|_{A_i \times B_i}, 
\end{eqnarray*}
where $A_1 = A$, $B_1 = B$, and for $i > 1: A_i = A \setminus \bigcup_{1 \le j < i} A(M_j)$ and $B_i = B(M_{j-1})$.
\end{lemma}

\begin{proof}
 The proof is similar to the proof of Lemma~\ref{lemma:no-length-2-deg-min-path}. The matchings 
$M_1, \dots, M_{d^*}$ can be obtained as follows. For each $b \in B(S)$, label its incident
edges in $S$ by $1, 2, \dots, \deg_{S^*}(b)$ arbitrarily. Matching $M_i$ is then the subset of edges of $S$ that are labeled
by $i$. 

We prove the statement by contradiction. Let $i$ be the smallest index such that $M_i$ is not a maximum matching in 
$G|_{A_i \times B_i}$. Then there exists an augmenting path $A = (a_1, b_1, \dots a_l, b_l)$ such that for all
$j < l: (a_{j+1}, b_j) \in M_i$ and $\forall i: (a_i, b_i) \notin M_i$. Let $b'$ be the match of $a_1$ in $S^*$. 
Since $a_1 \in A_l$, $\deg_{S^*}(b') > i$. Since $b_l \in B_i$ and $b_l$ is not matched in $M^*_i$, 
$\deg_{S^*}(b_l) < i$. Then $P = (b', a_1, b_1, \dots, a_l, b_l)$ is a degree-minimizing path contradicting our
assumption. \qed
\end{proof}
%\begin{lemma} \label{lemma:matching-half-of-the-nodes}
% Let $A' \subseteq A$, let $S = \semi_2(A', B, E)$ be a semi-matching in $G|_{A' \times B}$ that does not admit length $2$ 
% degree-minimizing paths, and let $S^* = \semi(A', B, E)$ be an optimal semi-matching in $G|_{A' \times B}$.
%Then $\exists A'' \subseteq A'$ with $|A''| \ge 1/2 |A'|$ such that
%\begin{enumerate}
% \item $\degmax S|_{A'' \times B} \le \degmax S^*$, \label{item:01}
% \item $S|_{A' \setminus A'' \times B}$ is a semi-matching of $G|_{A' \setminus A'' \times B}$ and it
% does not admit length $2$ degree-minimizing paths.  \label{item:02}
%\end{enumerate}
%\end{lemma}
%We are now ready to state Theorem~\ref{theorem:log-n-approximation}.

We firstly prove a lemma that is required in the proof of Theorem~\ref{theorem:log-n-approximation}.

\begin{lemma} \label{lemma:matching-half-of-the-nodes}
  Let $A' \subseteq A$, let $S = \semi_2(A', B, E)$ be a semi-matching in $G|_{A' \times B}$ that does not admit length $2$ 
 degree-minimizing paths and let $S^* = \semi(A', B, E)$ be an optimal semi-matching in $G|_{A' \times B}$.
Then $\exists A'' \subseteq A'$ with $|A''| \ge 1/2 |A'|$ such that
\begin{enumerate}
 \item $\degmax S|_{A'' \times B} \le \degmax S^*$, \label{item:01}
 \item $S|_{A' \setminus A'' \times B}$ is a semi-matching of $G|_{A' \setminus A'' \times B}$ and it
 does not admit length $2$ degree-minimizing paths.  \label{item:02}
\end{enumerate}
\end{lemma}

\begin{proof}
Let $d = \degmax S$ and let $d^* = \degmax S^*$. Partition $S$ into matchings $M_1, \dots, M_d$ as in Lemma~\ref{lemma:no-length-2-deg-min-path}.
We will show that $A'' = \bigcup_{i \le d^*} A(M_i)$ fulfills Item~\ref{item:01} and Item~\ref{item:02} of the Lemma. 

We have to show that $|A''| \ge 1/2 |A'|$. Let $A''' = A' \setminus A''$ and let $(a, b) \in S^*$ be an edge such that
$a \in A'''$. We argue now, that $\deg_S(b) \ge d^*$. 

Suppose for the sake of a contradiction that  $\deg_S(b) < d^*$. Then
$(a, b)$ could have been added to some matching $M_j$ with $j \le d^*$. Since by Lemma~\ref{lemma:no-length-2-deg-min-path} 
all $M_i$ are maximal, we obtain a contradiction and this proves that $\deg_S(b) \ge d^*$. 

This implies further that 
$|A''| \ge d^* \cdot | B(S^*|_{A''' \times B}) | \ge d^* \cdot |A'''|/d^* = |A'''|$,
where the last inequality comes from the fact that a node $b \in B(S^*|_{A''' \times B})$ has at most $d^*$ edges incident in $S^*$.
Since $A'''$ and $A''$ form a partition of $A'$, we obtain $|A''| \ge 1/2 |A'|$.

Since $A'' = A(S|_{A'' \times B})$ and $S|_{A'' \times B}$ is a set of $d^*$ matchings, Item~\ref{item:01} is trivially true.
Concerning Item~\ref{item:02}, note that if $S|_{A' \setminus A'' \times B}$ admitted a length $2$ degree-minimizing path, then that path
would also be a degree-minimizing path in $S$ contradicting the premise that $S$ does not admit a length $2$ degree-minimizing path. \qed
\end{proof}

\begin{theorem} \label{theorem:log-n-approximation}
 Let $S = \semi_2(A, B, E)$ be a semi-matching of $G$ that does not admit a length $2$ degree-minimizing path. Let $S^*$ be an optimal 
semi-matching in $G$. Then:
\begin{eqnarray*}
 \degmax S \le \lceil \log(n + 1) \rceil \degmax S^*.
\end{eqnarray*}
\end{theorem}

\begin{proof}
 We construct a sequence of vertex sets $(A_i)$ and a sequence of semi-matchings $(S_i)$ as follows. 
 Let $A_1 = A$, and let $S_1 = S$. For any $i$, $S_i$ will be a semi-matching in the graph $G|_{A_i \times B}$ and 
it will not admit length $2$ degree-minimizing paths. 

We construct $A_{i+1}$ and $S_{i+1}$  from $A_i$ and $S_i$ as follows. By Item~\ref{item:01} of Lemma~\ref{lemma:matching-half-of-the-nodes},
 there is a subset $A'_i \subseteq A_i$ of size at least $1/2 |A_i|$ such that $S_i|_{A'_i \times B}$ has maximal degree
 $d^*$. Let $A_{i+1} = A_i \setminus A'_i$, and let $S_{i+1} = S_i|_{A_{i+1} \times B}$. By Item~\ref{item:02} of Lemma~\ref{lemma:matching-half-of-the-nodes}, 
$S_{i+1}$ does not comprise length $2$ degree-minimizing paths in the graph $G|_{A_{i+1} \times B}$. We stop this construction at iteration
$l$ when $A'_l = A_l$ occurs. 

Note that $S = \bigcup_i S_i|_{A'_i \times B}$ and hence $\degmax S \le \sum_{i=1}^ l \degmax S_i|_{A'_i \times B} \le l \cdot d^*$.
It remains to argue that $l \le \log(n) + 1$. Since $|A'_i| \ge 1/2 |A_i|$ and $A_{i+1} = A_i \setminus A'_i$, we 
have $|A_{i+1}| \le 1/2 |A_i|$. Since $|A_1| = n$, we have $|A_i| \le (\frac{1}{2})^{i-1} n$. Then, 
$|A_{\lceil \log(n + 1) \rceil}| < 1$ which implies that $|A_{\lceil \log(n + 1) \rceil}| = 0$. 
We obtain hence $l \le \lceil \log(n + 1) \rceil$, which proves the theorem. \qed
\end{proof}

\bibliography{semi-matching-long}

\end{document}